\documentclass[11pt,a4paper,amssymb,amsmath, tightenlines]{article}

\usepackage{amsmath, amssymb, amsthm, latexsym, mathrsfs, url, color, epsfig, graphics, float, setspace,hyperref}
\usepackage{sectsty}
\usepackage[T1]{fontenc}
\usepackage[bitstream-charter]{mathdesign}
\usepackage{pgfplots}
\newtheorem{Theorem}{Theorem}[section]
\newtheorem{Definition}{Definition}[section]
\newtheorem{Remark}{Remark}
\newtheorem{Proposition}{Proposition}[section]
\newtheorem{Lemma}{Lemma}[section]
\newtheorem{Corollary}{Corollary}[section]

\numberwithin{equation}{section}

\newcommand{\e}{{\rm e}}
\newcommand{\E}{\mathbb{E}}

\newcommand{\F}{\mathcal{F}}

\newcommand{\PR}{\mathbb{P}}

\newcommand{\Q}{\mathbb{Q}}
\newcommand{\R}{\mathbb{R}}

\newcommand{\rd}{\textup{d}}

\newcommand{\indi}[1]{1\hspace{-.09cm}\textup{\textrm{l}}}

\newcommand{\nn}{\nonumber}

\begin{document}
\title{\vspace{-2cm}\bf Consistent Valuation Across Curves \\ Using Pricing Kernels}
\author{Andrea Macrina{$^{\dag\, \ddag}$\footnote{Corresponding author: a.macrina@ucl.ac.uk}\ }, Obeid Mahomed{$^{\S}$} \\ \\ {$^{\dag}$Department of Mathematics, University College London} \\ {London WC1E 6BT, United Kingdom} \\ \\ {$^{\ddag}$Department of Actuarial Science, University of Cape Town} \\ {Rondebosch 7701, South Africa} \\ \\ {$^{\S}$African Institute of Financial Markets and Risk Management} \\ {University of Cape Town, Rondebosch 7701, South Africa}}
\date{\today}

\maketitle
\vspace{-0.75cm}
\begin{abstract}
\noindent The general problem of asset pricing when the discount rate differs from the rate at which an asset's cash flows accrue is considered. A pricing kernel framework is used to model an economy that is segmented into distinct markets, each identified by a yield curve having its own market, credit and liquidity risk characteristics. The proposed framework precludes arbitrage within each market, while the definition of a curve-conversion factor process links all markets in a consistent arbitrage-free manner. A pricing formula is then derived, referred to as the across-curve pricing formula, which enables consistent valuation and hedging of financial instruments across curves (and markets). As a natural application, a consistent multi-curve framework is formulated for emerging and developed inter-bank swap markets, which highlights an important dual feature of the curve-conversion factor process. Given this multi-curve framework, existing multi-curve approaches based on HJM and rational pricing kernel models are recovered, reviewed and generalised, and single-curve models extended. In another application, inflation-linked, currency-based, and fixed-income hybrid securities are shown to be consistently valued using the across-curve valuation method.
\\\vspace{-0.2cm}
\\\vspace{-0.2cm}\\
\end{abstract}

\section{Introduction}

The fundamental problem that we consider is the valuation of a financial instrument using a discounting rate which differs from the rate at which the instrument's future cash flows accrue. Since such financial instruments are synonymous with fixed income assets, we will focus thereon. Nonetheless, we have no reason to believe that the framework that we develop cannot be extended to the valuation of a generic financial asset. The financial crisis brought this valuation problem to the foreground when substantial spreads emerged between inter-bank interest rates that were previously bound by single yield curve consistencies, culminating in a new valuation paradigm of multiple yield curves---one used for discounting (the overnight indexed swap (OIS) yield curve) and others used for forecasting of cash flows (the $y$-month inter-bank offered rate (IBOR) curve, $y = 1,3,6,12$). However, this problem was also prevalent pre-crisis when an economy is considered, which has an inter-bank swap market, a government bond market and trades in the global economy via the foreign exchange market, resulting in three different curves: the nominal swap curve, the government bond curve and the foreign exchange basis curve. The valuation of any financial instrument that is issued in one of these markets, but has cash flows that are determined by any of the other markets, once again manifests the fundamental problem.
\

In this paper, we directly address the fundamental problem, articulated above, in a general sense. Considering the aftermath of the financial crisis however, academic literature on multi-curve interest rate modelling (in the context of the \textit{developed} inter-bank swap market) has evolved rapidly. Here we classify this literature into four categories or modelling approaches and provide a non-exhaustive list of references and a brief summary of the main contributions therein.
\

The first category is \textit{short-rate models}. Kijima et al. \cite{ktw} propose a three-yield curve model (discount, swap and government bond curve) for an economy with the respective short-rates governed by Gaussian, exponentially quadratic models. Kenyon \cite{k} and Morini \& Runggaldier \cite{mr} consider Vasicek, Hull-White (HW) and Cox-Ingersoll-Ross (CIR) short-rate models for the OIS, IBOR and/or OIS-IBOR spread curves. Filipovi\'c \& Trolle \cite{ft} propose a Vasicek process with stochastic long-term mean as the OIS short-rate model with explicit models for default and liquidity risk. Alfeus et al. \cite{ags} adopt a novel approach of modelling ``roll-over risk'' explicitly in a reduced-form setting and consider multi-factor CIR-type processes for this and the OIS short-rate.
\

\textit{Heath-Jarrow-Morton (HJM) models} constitute the second category. Pallavicini \& Tarenghi \cite{pt}, Fujii et al. \cite{fst}, Moreni \& Pallavicini \cite{mp}, Cr\'epey et al. \cite{cgns} and Miglietta \cite{m} all focus on a \textit{hybrid HJM-LMM (LIBOR Market Model) approach} where the OIS curve is modelled using the classical HJM model, while the IBOR forward rates are modelled in an ad hoc manner. Cr\'epey et al. \cite{cgn} pioneered the use of the HJM framework via a credit risk analogy, while Miglietta \cite{m} and Grbac \& Runggaldier \cite{gr} do the same using a foreign exchange (FX) analogy. Pallavicini and Tarenghi \cite{pt} focus on aspects of calibration, while Moreni \& Pallavicini \cite{mp} propose a specific Markovian factor representation which expedites calibration. Cr\'epey et al. \cite{cgns} consider L\'evy driven models, while Cuchiero et al. \cite{cfg} consider a general semimartingale setup with multiplicative OIS-IBOR spreads.
\

Category three is the class of \textit{LIBOR Market Models (LMMs)}. Morini \cite{mm}, Mercurio \cite{fm}, \cite{fm1}, \cite{fm2}, \cite{merc} and Bianchetti \cite{bianchetti} were the first to extend the LMM to a multi-curve setting, with the latter doing so via an FX analogy. Mercurio \cite{fm1} and Mercurio and Xie \cite{mx} formalised the first approach utilising an additive spread between OIS and IBOR forward rates that were modelled as martingales under the classical forward measure, while Ametrano and Bianchetti \cite{ab} formalised the associated multi-curve bootstrapping process. Grbac et al. \cite{gpss} provide an alternative to the aforementioned approach using a class of affine LIBOR models, first proposed by Keller-Ressel et al. \cite{kpt}.
\

The fourth and final catergory are \textit{pricing kernel models}. At the present time, we are only aware of Cr\'epey et al. \cite{cmns} and Nguyen \& Seifried \cite{ns} who have formulated multi-curve systems with pricing kernels. We highlight here that, in our opinion, both these papers adopt a hybrid pricing kernel-LMM approach since the OIS curve is modelled with a pricing kernel while the IBOR process is modelled in an ad hoc fashion---we will expand on this in Sections \ref{PK-Multi} and \ref{R-Sec}. In this paper we develop a pure pricing-kernel based approach, which we believe to be the first of such a modelling class.
\

For a detailed review of the post-crisis multi-curve interest rate paradigm from both, a theoretical and practical perspective, we refer the reader to Bianchetti \& Morini \cite{mbmm}, Grbac \& Runggaldier \cite{gr} and Henrard \cite{Henrard}.
\

The solution that we propose rests upon a pricing formula, which we call the {\it across-curve pricing formula}. This formula has a pricing kernel-based model for the economy as its foundation. More specifically, the pricing kernel framework models the set of yield curves associated with the respective economy under consideration. This enables us to link the set of yield curves in a consistent arbitrage-free manner through the definition of a {\it curve-conversion factor process}. This conversion process plays an important dual role, giving rise to the \textit{across-curve pricing formula} that enables consistent valuation and hedging of financial instruments across curves. It turns out that the \textit{curve-conversion factor process} is consistent with an FX process in multi-currency modelling in a pricing kernel framework, and therefore our approach is also consistent with the FX analogy first proposed by Bianchetti \cite{bianchetti} for interest rate derivatives (or the \textit{developed} inter-bank swap market) in an LMM setting. In our work we are interested in more than \textit{developed markets} or the inter-bank swap market, and endeavour to build consistent relations among a wide variety of \textit{developed} and \textit{emerging market} fixed-income assets including inflation-linked notes, FX contracts, and hybrid products such as inflation-linked FX instruments. Here, we mention Flesaker \& Hughston \cite{FH1b} and \cite{FH2} for an argitrage-free pricing kernel approach to the valuation of FX securities, and to Frey \& Sommer \cite{fs} if one were to consider extending classical short rate models, based on diffusion processes with deterministic coefficients, for FX-rates. The approach by Jarrow \& Yildirim \cite{JY}, based on the HJM-framework, might be treated as in Section \ref{xy-HJM-Sec} and used for inflation-linked pricing as shown in Section \ref{IL-FX}, later in this paper. We note here the early work in 1998 by Hughston \cite{Hughston98} who produced a general arbitrage-free approach to the pricing of inflation derivatives, in which---to our knowledge---a foreign exchange analogy was used in such a context, for the first time. In Hughston's work, the CPI is treated like a foreign exchange rate that links the nominal and the real price systems as if they were domestic and foreign currencies, respectively.  The work by Pilz \& Schl\"ogl \cite{PS} on modelling commodity prices re-interprets a multi-currency LMM approach. Similarities can be seen when applying our approach to multi-currency and multi-curve LIBOR models, as developed in Section \ref{MCIRFX}, where an FX-LIBOR forward rate agreement is priced. In all that follows, we refer to the discounting curve as the $x$-curve and the forecasting curve as the $y$-curve. Therefore, when describing our framework, we speak of the \textit{xy-formalism}, while we refer to the application thereof as the \textit{xy-approach}.
\

With regard to the multi-curve system adopted by \textit{developed market} practitioners for their inter-bank swap market, we will show that there is a natural formulation of such a system within our framework. Moreover, we will show that this natural formalism is not adopted by practitioners, or the market in a strict sense in general. Rather, practitioners have adopted a more rigid version of the flexible multi-curve system we propose, the choice of which results and ensures simpler specifications for fundamental interest rate products, i.e. forward rate agreements (FRAs) and interest rate swaps (IRSs). We also formulate a multi-curve system for \textit{emerging markets}, one that is remarkably consistent with the corresponding \textit{developed market} system---this feature being entirely attributable to the critical dual role played by the \textit{curve-conversion factor process}.  We will expand on this in Sections \ref{CCF} and \ref{PK-Multi}.

The remainder of this paper is structured as follows: Section \ref{CCF} introduces the \textit{curve-conversion factor process} and the \textit{across-curve pricing formula}. Section \ref{PK-Multi} introduces consistent multi-curve interest rate systems for \textit{developed} and \textit{emerging} inter-bank swap markets. Section \ref{xy-HJM-Sec} reviews and reformulates existing HJM multi-curve modelling approaches within the context of the xy-approach, and introduces a new framework that we call the xy-HJM framework. Section \ref{R-Sec} introduces a generic class of rational multi-curve models and revisits recent rational multi-curve approaches based on pricing kernels in light of the xy-framework. Moreover, the linear-rational term structure models are shown to belong to a more general class of pricing-kernel-based (rational) models and are extended to the multi-curve setup. In Section \ref{IL-FX}, the across-curve pricing approach is adopted to price inflation-linked and FX securities, including hybrid contracts. In Section \ref{Concl} we draw various conclusions and take the opportunity to summarise our findings.

\section{Across-curve pricing formula}\label{CCF}

In this section we define the \textit{curve-conversion factor process} and deduce what we term the \textit{across-curve pricing formula}. At the basis of the curve-conversion factor process lies the assumption that, within a given economy, there is a distinct market associated with each curve. Each of these markets are characterised by its own set of market, liquidity and credit risk factors. In turn, each set of market, liquidity and credit risk factors may be systematic or idiosyncratic in nature. The curve-conversion factor process plays a dual role: (i) it provides a mechanism---akin to a ladder---that enables one to transit consistently from one discount curve system to another; and (ii) it facilitates the equivalent representation of cash flows across markets (or curves), no matter what financial instrument is implicitly being priced or interest rate system being modelled. This feature enables consistent valuation across different curves (or markets). The paradigm we shall adopt for the development of the across-curve pricing approach is one based on pricing kernels. Previous works developing and applying the pricing kernel paradigm comprise, e.g., Constantinides \cite{const}, Flesaker \& Hughston \cite{FH1} and \cite{FH2}, Rogers \cite{LCGR}, Jin \& Glasserman \cite{JG}, Hughston \& Rafailidis \cite{HR}, Akahori et al. \cite{ahtt}, Macrina \cite{am}, and Filipovi\'c et al. \cite{flt}. Next, we introduce the stochastic basis and the pricing kernel system.
\

We consider a filtered probability space $(\Omega, \F, (\F_t)_{t\ge 0},\PR)$, where $(\F_t)_{t\ge 0}$ denotes the filtration and $\PR$ the real-world probability measure. We introduce an $(\F_t)$-adapted pricing kernel process $(h_t)_{t\ge 0}$, which governs the inter-temporal relation between asset values at different times in a financial market. It is a fundamental ingredient   in the so-called standard no-arbitrage pricing formula, for a non-dividend-paying financial asset $H$, given by
\begin{equation}\label{stdpricingformula}
H_{t}=\frac{1}{h_t}\E\left[h_T\, H_{T}\,\vert\,\F_t\right].
\end{equation}
The no-arbitrage asset price process $(H_t)_{t\ge 0}$ is obtained by taking the conditional expectation of the random cash flow $H_T$, occurring at the fixed future date $T\ge t\ge 0$, that is discounted by the pricing kernel. Standard references, in which asset pricing using pricing kernels is discussed, include, e.g., Hunt \& Kennedy \cite{HKen}, Duffie \cite{duff}, Cochrane \cite{coch}, and  Grbac \& Runggaldier \cite{gr}.
\

In order for us to deduce the across-curve pricing formula---seen as an extension to the pricing formula (\ref{stdpricingformula})---we assume the existence of a set of (continuous-time) $(\F_t)$-adapted pricing kernel processes $(h^{y}_t)_{t\ge 0}$, where $y=0, 1, 2,\ldots,n$, each linked to a distinct $y$-market. The price $H^y_{t}$ at $t\in[0,T]$ of a non-dividend-paying financial asset $H$, with (random) cash flow $H^y_{T}$ at the fixed future date $T$, is then given by
\begin{equation}\label{y-stdpricingformula}
H^y_{t}=\frac{1}{h^{y}_t}\E\left[h^{y}_T\, H^y_{T}\,\vert\,\F_t\right].
\end{equation}
The superscript $y$ emphasises that the pricing formula (\ref{y-stdpricingformula}) holds for the valuation of assets in the $y$-economy (or in the $y$-market). In fact, the pricing kernel process $(h^{y}_t)$ governs the inter-temporal relation between the present value of financial assets and their future cash flows in the associated $y$-economy. It then follows in a straightforward manner, that the price process $(P^{y}_{tT})_{0\le t\le T}$ of a zero coupon bond (ZCB), with payoff $H^y_{T}=P^y_{TT}=1$ at the fixed maturity $T$ and quoted in the $y$-market, is given by
\[
P^{y}_{tT}=\frac{1}{h^{y}_t}\E\left[h^{y}_T\,\vert\,\F_t\right].
\]
The discount bond system---spanned in theory by a continuum, but in practice a finite number of maturities $T=T_1, T_2,\ldots, T_n$---generates a term structure curve. Since this curve is indexed by the particular market $y$, we refer to it as the {\it $y$-curve}. In all that follows, we single out one of the set of the $y$-markets (and thereby its associated $y$-curve) and refer to it as the $x$-market (and its associated term structure curve as the $x$-curve); of course then this market also has an associated $(\F_t)$-adapted pricing kernel $(h^{x}_t)$. The $x$-market is the market within which pricing (or discounting) occurs, while the $y$-market will denote the market within which the cash flows of the financial instruments are forecasted (or accrued).
\

The fundamental pricing problem that is considered in this paper is one where a financial instrument accrues cash flows at a rate of interest that differs from that used for discounting. First we consider the problem of cash flow forecasting and equivalent representation under different curves (or markets), before we tackle the problem of valuation (or discounting). An equivalent cash flow representation across curves (or markets) is justified in Appendix \ref{App-ArbStrat-CC} using no-arbitrage portfolio-based strategies. These findings are formalised in the following definition that introduces the {\it curve-conversion factor process}.

\begin{Definition}\label{Qxy-Prop}
Consider an economy with $n$ distinct markets characterised by a set of pricing kernel processes $(h^{y}_t)$ and associated discount bond systems $(P^y_{tT})$, where $y=0, 1, 2,\ldots,n$ and $0 \leq t \leq T$. The converted value $C^x_t$ in the $x$-market at time $t$ of any spot cash flow $C_t^y$ determined in the $y$-market is given by
\begin{equation*}
C_t^x = \frac{h^y_t}{h^x_t} C_t^y \,,
\end{equation*}
where $x,y = 0,1,\ldots,n$. The converted value $C_{t}^x(t,T)$ at time $t$ in the $x$-market of any forward cash flow $C_{t}^y$(t,T), measurable at time $t$ but payable at time $T$, determined in the $y$-market is given by
\begin{equation*}
C_{t}^x(t,T) = \frac{h^y_t P^y_{tT}}{h^x_t P^x_{tT}} C_{t}^y(t,T) \,,
\end{equation*}
where $x,y = 0,1,\ldots,n$. These two relations are combined by the definition of the $(\F_t)$-adapted curve-conversion factor process
\begin{equation}\label{xy-conv-factor}
Q^{xy}_{tT}=\frac{\E\left[h^{y}_T\,\vert\,\F_t\right]}{\E\left[h^{x}_T\,\vert\,\F_t\right]}=\frac{h^y_t P^y_{tT}}{h^x_t P^x_{tT}},
\end{equation}
where $t\in[0,T]$ is the time until which the cash flow being converted is measurable and $T>0$ is the payment date.
\end{Definition}
We note that the cash flows $C_{t}^x(t,T)$ and $C_{t}^y(t,T)$ are linked by the identity $C_{t}^x(t,T)=Q^{xy}_{tT}C_{t}^y(t,T)$, for $t\in[0,T]$. With this definition at hand, we now have the necessary tool to resolve the fundamental pricing problem considered in this paper, i.e. valuing a generic financial instrument that accrues cash flows under one curve, the $y$-curve, but is priced under another curve, the $x$-curve. Our approach is consistent with the FX analogy proposed by Bianchetti \cite{bianchetti}, but formalised in an economy modelled by a set of pricing kernels---we describe our approach as the xy-formalism. At the heart of this formalism is the pricing formula presented next. We refer to this formula as the \textit{across-curve pricing formula}. The relation of this novel formula to the fundamental pricing formula (\ref{stdpricingformula}) is shown in the proof of the following proposition.
\begin{Proposition}\label{XY-gen-Prop}
Let $0\le s\le t\le T$. Consider a generic financial asset $H$ that has a single $\F_t$-measurable cash flow $H^y_t(t,T)$ occurring at the fixed time $T\ge t$ and determined by the $y$-curve (or the $y$-market). It is noted that in the time interval $[t,T]$, the quantity $H^y_t(t,T)$ is fixed at the value observed at time $t\ge 0$. Within the $xy$-approach, the price process $(H^{xy}_{sT})_{0\le s\le T}$ of a financial instrument, determined by the $x$-curve (or $x$-market) and contingent on the asset $H$, is given by 
\begin{equation}\label{xy-generic}
H^{xy}_{sT}=
\begin{cases}
\displaystyle\frac{1}{h^x_s}\E\left[h^x_t\,P^x_{tT}\,Q^{xy}_{tT}\,H^y_t(t,T)\,\vert\,\F_s\right], &\quad 0\le s<t,\\
\\
P^x_{sT}\,Q^{xy}_{tT}\,H^y_t(t,T), &\quad t\le s\le T.
\end{cases}
\end{equation}
The curve-conversion factor process $(Q^{xy}_{tT})_{0\le t\le T}$ is introduced in Definition \ref{Qxy-Prop}.
\end{Proposition}
%
\begin{proof}
For information we note that, by an application of the relation (\ref{y-stdpricingformula}), the price process $(H^y_t)_{0\le t \le T}$ of the financial asset $H$ is deduced to be 
\begin{equation}
H^y_t=\frac{1}{h^y_t}\E\left[h^y_T\,H^y_t(t,T)\,\vert\,\F_t\right]=H^y_t(t,T)P^y_{tT},
\end{equation} 
since the cash flow $H^y_t(t,T)$ is $\F_t$-measurable and it occurs at $T$. At time $t\in[0,T]$, we convert $H^y_t(t,T)$ to the corresponding value $H^x_t(t,T)$ in the $x$-market by use of the conversion factor $Q^{xy}_{tT}$:
\begin{equation}
H^x_t(t,T)=Q^{xy}_{tT}H^y_t(t,T).
\end{equation}
Now we insert the converted cash flow $H^x_t(t,T)$ in the standard no-arbitrage formula (\ref{stdpricingformula}) (or formula (\ref{y-stdpricingformula}), where y=0 is taken to be the $x$-curve) where $h_t=h^x_t$ is assumed. We have,
\begin{equation}
H^x_{sT}=\frac{1}{h^x_s}\E\left[h^x_t\,H^x_t(t,T)\,\vert\,\F_s\right].
\end{equation}
Given that $H^x_t(t,T)$ is $\F_t$-measurable, we deduce the following by the tower property of conditional expectation:
\begin{equation}
H^x_{sT}=\frac{1}{h^x_s}\E\left[\E\left[h^x_T\,\vert\,\F_t\right]H^x_t(t,T)\,\vert\,\F_s\right]=\frac{1}{h^x_s}\E\left[h^x_t\,P^x_{tT}\,H^x_t(t,T)\,\vert\,\F_s\right],
\end{equation}
for $0\le s<t$. In addition, for $t\le s\le T$, we have $H^x_{sT}=P^x_{sT}\,H^x_t(t,T)$. Recalling that $H^x_t(t,T)=Q^{xy}_{tT}\,H^y_t(t,T)$, and by choosing to write $H^{xy}_{sT}$ for $H^{x}_{sT}$ in order to emphasise the interaction between the $x$- and the $y$-curves, the proof is complete. We add that the one-to-one across-curve extension to the standard pricing formula (\ref{stdpricingformula}) is recovered by setting $t=T$ in the relation (\ref{xy-generic}).
\end{proof}
%

\begin{Remark}\label{Rem-xyZCB}
When $t = T$ and $H^y_{TT} = 1$, using Proposition \ref{XY-gen-Prop}, we may define the ZCB
\begin{equation}\label{xy-ZCB}
P^{xy}_{sT} = \frac{1}{h^{x}_s}\E\left[h^{y}_T\,\vert\,\F_s\right] = Q_{ss}^{xy} P^y_{sT} = P^x_{sT} Q_{sT}^{xy}\,,
\end{equation}
for $s \in [0,T]$, which has two representations using the definition of the conversion factor (\ref{xy-conv-factor}).
\end{Remark}

Given Proposition \ref{XY-gen-Prop}, we can now present the dual role played by the \textit{curve-conversion factor process}, within the xy-formalism, which is described in the following corollary.

\begin{Corollary}\label{Qxy-dual-coro}
Within the xy-formalism, if the cash flow $H^y_t(t,T)$ is directly observable in the economy, then the curve-conversion factor process enables valuation by acting at the level of the discounting curve as follows:
\begin{equation}\label{y-obs-payoff}
H^{xy}_{sT}=\frac{1}{h^{x}_s}\E\left[h^{x}_t P^x_{tT}Q^{xy}_{tT} H^{y}_{t}(t,T)\,\vert\,\F_s\right] = \frac{1}{h^{x}_s}\E\left[h^{y}_t P^y_{tT} H^{y}_{t}(t,T)\,\vert\,\F_s\right].
\end{equation}
However, if the curve-converted cash flow $H^{xy}_{tT}$ is directly observable in the economy, then the curve-conversion factor process enables valuation by acting at the level of the cash flow as follows:
\begin{equation}\label{xy-obs-payoff}
H^{xy}_{sT}=\frac{1}{h^{x}_s}\E\left[h^{x}_t P^x_{tT}Q^{xy}_{tT} H^{y}_{t}(t,T)\,\vert\,\F_s\right] = \frac{1}{h^{x}_s}\E\left[h^{x}_t H^{xy}_{tT}\,\vert\,\F_s\right],
\end{equation}
where $(H^{xy}_{sT})_{0\le s \le t}$ is the $x$-market value of $H^y_t(t,T)$, for $s \leq t \leq T$.
\end{Corollary}
\begin{proof}
If $H^y_t(t,T)$ is determined in the $y$-market and directly observable (i.e. quoted) within the economy, then according to Proposition \ref{XY-gen-Prop} the value of such a payoff within the $x$-market, at the future terminal time $T$, is given by
\begin{equation}
H^{xy}_{TT} = Q^{xy}_{tT}H^y_t(t,T),
\end{equation}
which is model-implied, since the curve-conversion factor process $Q^{xy}_{tT}$ is determined by the specific forms of the pricing kernels $(h^x_t)$ and $(h^y_t)$, respectively. Therefore, since $H^{xy}_{TT}$ is not directly observable in the economy due to $Q^{xy}_{tT}$, the curve-conversion factor process is subsumed into the discounting process in Eq. (\ref{xy-generic}) for $0 \le s < t$, by observing that $h^x_t P^x_{tT} Q^{xy}_{tT} = h^y_t P^y_{tT}$, which yields Eq. (\ref{y-obs-payoff}).\\
Conversely, if $H^y_t(t,T)$ is determined in the $y$-market but the converted quantity $H^{xy}_{TT}$ is directly observable within the economy, then
\begin{equation}
H^y_t(t,T) = \frac{H^{xy}_{TT}}{Q^{xy}_{tT}}, 
\end{equation}
is model-implied, which is subsumed into the cash flow process by observing that $H^{xy}_{sT} = P^x_{sT}Q^{xy}_{tT}H^y_t(t,T)$ for $t \le s \le T$ from Eq. (\ref{xy-generic}), which yields Eq. (\ref{xy-obs-payoff}).
\end{proof}

\begin{Remark}
Corollary \ref{Qxy-dual-coro} proves to be critical in Section \ref{PK-Multi}, where consistent mutli-curve systems are derived for both, developed and emerging inter-bank swap markets. With regard to FRAs (the fundamental inter-bank swap market derivative), which has an IBOR process as its underlying, it turns out that the $y$-market determined IBOR process is directly observable in the emerging market, but its curve-converted equivalent is directly observable in the developed market. In this instance, the dual nature of the curve-conversion factor process caters for this apparent cross-economy market inconsistency, resulting in one consistent modelling framework.
\end{Remark}
In Appendix \ref{App-B}, we provide the consistent set of changes of numeraire assets and associated equivalent probability measures, which ensure that no arbitrage is produced when the across-pricing formula is applied using an equivalent martingale measure.

\section{Pricing kernel approach to multi-curve systems}\label{PK-Multi}

First we consider the definition of a spot IBOR, i.e. a deposit rate that is offered at a fixed time $t \geq 0$ by a set of suitably credit-rated banks within a given economy. We assume that the maturity of said IBOR is $t+\delta > t$, so that the associated tenor is given by $\delta >0$. Then we may define (or represent) the spot IBOR process via ZCB instruments by 
\begin{equation}\label{spot-IBOR}
L_t(t,t+\delta) = \frac{1}{\delta} \left(\frac{1}{P_{tt+\delta}} - 1\right),
\end{equation}
where $t \geq 0$, $\delta >0$, and where $P_{tt+\delta}$ is the price at time $t$ of a ZCB, with tenor $\delta$, that matures at time $t+\delta$. In the classical single-curve framework, where IBORs are considered an appropriate proxy for risk-free rates and where a tradable discount bond system is assumed, one can then proceed to define the forward IBOR process via the canonical no-arbitrage pricing relation
\begin{equation}\label{mart-IBOR}
L_t(T_{i-1},T_i)=\frac{1}{h_t P_{tT_i}} \E\left[h_{T_{i-1}} P_{T_{i-1}T_i} L_{T_{i-1}}(T_{i-1},T_i) \, \big\vert \, \F_t \right],
\end{equation}
for $0\le t\le T_{i-1}$, and where $\delta_i=T_i-T_{i-1}$ is the IBOR tenor and $(h_t)_{t \geq 0}$ is the pricing kernel process. By use of the relation (\ref{spot-IBOR}) with $t=T_{i-1}$ and $\delta = \delta_i$, and the ZCB pricing relation $h_t P_{tT_i}=\E[ h_{T_{i-1}} P_{T_{i-1}T_i}\, \vert \, \F_t]$, one obtains the forward IBOR process
\begin{equation}\label{class-IBOR}
L_t(T_{i-1},T_i)= \frac{1}{\delta_i}\left(\frac{P_{tT_{i-1}}}{P_{tT_i}}-1\right),
\end{equation}
for $0\le t\le T_{i-1}$. We note that the product of the pricing kernel process and the discounted forward IBOR process $(h_t P_{tT_i} L_t(T_{i-1},T_i))_{0 \leq t \leq T_{i-1}}$ is an $((\F_t),\PR)$-martingale, which is analogous to the forward IBOR process being a martingale under the $\Q^{T_i}$-forward measure in the classical single-curve theory.
\

The classical relation (\ref{class-IBOR}) states that the forward IBOR value at time $t$ can be replicated by a linear combination of zero-coupon bonds, i.e. by one maturing at the IBOR reset date $T_{i-1}$ and another ZCB maturing at the IBOR settlement date $T_i$. In a market where the spread between an overnight indexed swap (OIS) rate and the corresponding IBOR is non-zero, relation (\ref{class-IBOR}) is no longer acceptable. That is, the now \textit{risky} IBOR can no longer be replicated using \textit{risk-free} ZCBs. In oder words, the IBOR market is exposed to risk factors which are not necessarily affecting the \textit{risk-free} ZCB market, while the risk exposure also varies depending on the IBOR tenor $\delta_i=T_i-T_{i-1}$ one is investing in. Hence, one needs to assume that holding a financial contract written on a 3-month IBOR exposes an investor to a different risk profile than when holding an instrument written on a 6-month IBOR. It follows that assuming \textit{risk-free} ZCBs can replicate the same risk exposures as contracts written on an IBOR is wrong because: (a) an IBOR may be subject to more risk sources than the \textit{risk-free} ZCBs; and (b) the number of risk factors affecting an IBOR contract may depend on the IBOR tenor.
\

We ask the following question: If one insisted on keeping the relation (\ref{class-IBOR}), albeit subject to modifications, how would one need to adjust---{\it in a consistent and arbitrage-free manner}---the relation between an IBOR model and the associated ZCBs in a multi-curve setup? It turns out that the answer is an extension based on the xy-formalism introduced above. Here is how we do it.
\

First, we consider a collection of interest rate curves indexed by $x, y = 0,1,2,...,n$ where we refer to the $x$-curve as the \textit{discounting curve} and the $y$-curve as the \textit{forecasting curve}. An example for a pair of curves $(x,y)$ may be the pair $(0,1)$ where the 0-curve is the OIS curve and the 1-curve is the 1-month IBOR curve. The case where $x=y$ is the (classical) single-curve economy. Next we make the relationship (\ref{class-IBOR}) curve-dependent and write  
\begin{equation}\label{y-IBOR}
L^{y}_t(T_{i-1},T_i)=\frac{1}{\delta_i}\left(\frac{P^{y}_{tT_{i-1}}}{P^{y}_{tT_i}}-1\right).
\end{equation}
Thus, the $y$-ZCB system $P^{y}_{tT_i}$ has an associated $y$-tenored IBOR, which is subject to the same set of risk factors, i.e. the $y$-tenored IBOR defines 
the $y$-ZCB system. Moreover, the $y$-ZCB price process satisfies the martingale relation $h^{y}_t P^{y}_{tT_i} = \E[h^{y}_{T_i}\,\vert\,\F_t]$,
which is to say that no-arbitrage is assumed within the self-consistent $y$-market. 
\

Next we detail the development of consistent multi-curve interest rate systems inspired by the xy-formalism for both, emerging and developed markets.

\subsection{Discounting systems in emerging markets}\label{DS-EM}

In this section we consider the simpler case of an \textit{emerging market}, in particular one where no OIS zero-coupon yield curve exists. To be precise, the spot overnight rate is observable but there are no tradable and liquid overnight indexed swaps, i.e. there is no OIS derivative market to enable the construction of a yield curve. For more information on the specific nuances and issues relating to emerging inter-bank swap markets, we refer the reader to Jakarasi et al. \cite{jlm}, and references therein, who consider the problem of estimating an OIS zero-coupon yield curve in South Africa. In such a market, all forecasting and discounting of cash flows is done by one liquid, risky $y$-tenored IBOR zero-coupon yield curve, only. 

To derive the multi-curve discounting system within the xy-formalism, we first consider the pricing of standard FRAs. FRAs are the fundamental primitive securities in any interest rate market, which facilitate price discovery for forward IBORs.  The FRA considered here has reset time $T_{i-1}>0$ and maturity time $T_i>T_{i-1}$, which is also assumed to be the settlement time, and is therefore written on the future spot IBOR $L^y_{T_{i-1}}(T_{i-1},T_i)$. The value at time $t\in [0,T_i]$ of this FRA is denoted by $V^{yy}_{tT_i}$, with the first character of the superscript indicating the discount curve, and the second character denoting the forecasting curve. For a unit nominal, the FRA's payoff at $T_i$ is given by
\begin{equation}\label{yy-FRA-payoff}
V^{yy}_{T_iT_i} = \delta_i \left(L^{y}_{T_{i-1}}(T_{i-1},T_i)-K^y \right),
\end{equation}
where $K^y$ is an arbitrary strike rate expressed in the $y$-market. We emphasise that the FRA's payoff is actually measurable at time $T_{i-1}$, however the actual cash flow is only paid at time $T_i$.\footnote{The market convention is to understand the right-hand-side of (\ref{yy-FRA-payoff}) as the rate (or quote) observed at $T_{i-1}$ and applied at $T_i$ on one unit of currency giving the payoff value $V^{yy}_{T_iT_i}$ of the contract at $T_i$. Since this value is paid at $T_i$, we denote it $V^{yy}_{T_iT_i}$, and use the subscripts $T_i$.} As a consequence, we may also define the in-advance FRA payoff $V^{yy}_{T_{i-1}T_i}$ at $T_{i-1}$, which is the value $V^{yy}_{T_iT_i}$ discounted by $P^{y}_{T_{i-1}T_i}$, by
\begin{equation}\label{yy-FRA-advance-payoff}
V^{yy}_{T_{i-1}T_i} = P^{y}_{T_{i-1}T_i} \delta_i \left(L^{y}_{T_{i-1}}(T_{i-1},T_i)-K^y \right).
\end{equation}
Using the pricing formula (\ref{xy-generic}) with $x=y$, along with relations (\ref{mart-IBOR}), (\ref{class-IBOR}) and (\ref{y-IBOR}), the FRA price process is derived as
\begin{eqnarray}\label{yy-FRA}
V^{yy}_{tT_i}
=\frac{\delta_i}{h^{y}_t}\E\left[h^{y}_{T_{i-1}}V^{yy}_{T_{i-1}T_i}\,\big\vert\,\F_t\right]
&=&\frac{\delta_i}{h^{y}_t}\E\left[h^{y}_{T_{i-1}} P^{y}_{T_{i-1}T_i} \left(L^{y}_{T_{i-1}}(T_{i-1},T_i)-K^y \right)\,\big\vert\,\F_t\right]\nn\\
&=&\delta_i P^{y}_{tT_i}\left(L^{y}_t(T_{i-1},T_i)-K^y \right) \nn\\
&=& P^{y}_{tT_{i-1}} - (1+ \delta_i K^y)P^{y}_{tT_{i}}.
\end{eqnarray}
By setting $V^{yy}_{tT_i}=0$, the fair FRA rate process is recovered and is given by 
\begin{equation}\label{yy-fair-FRA-rate}
K^{yy}_{t}(T_{i-1},T_i)=L^{y}_t(T_{i-1},T_i),
\end{equation}
for $t \in [0,T_{i-1}]$. The notation $K^{yy}_{t}(T_{i-1},T_i)$ emphasises that this fair FRA strike rate applies when the $y$-curve is used for both, discounting and forecasting.
\

Next we consider a standard IRS with unit nominal, reset times $\{T_0,T_1,\ldots,T_{n-1}\}$, payment times $\{T_1,T_2,\ldots,T_n\}$, referencing the $y$-tenored IBOR and arbitrary fixed swap rate under the $y$-market denoted by $S^y$. Again applying pricing relation (\ref{xy-generic}) with $x=y$, together with relations (\ref{mart-IBOR}), (\ref{class-IBOR}) and (\ref{y-IBOR}), the IRS price process is derived as
\begin{eqnarray} \label{yy-IRS}
V^{yy}_{tT_n}
&=& \sum_{i=1}^n \frac{\delta_i}{h^{y}_t}\E\left[h^{y}_{T_i} \left(L^{y}_{T_{i-1}}(T_{i-1},T_i) - S^y \right) \,\big\vert\,\F_t\right] \nn\\
&=&P^{y}_{tT_{0}}  - P^{y}_{tT_{n}} - S^y \sum_{i=1}^n \delta_i P^{y}_{tT_{i}} ,
\end{eqnarray}
for $t \leq T_0$.  Using the same notation convention as with the FRA, the fair IRS rate process is given by
\begin{equation}\label{yy-fair-IRS-rate-process}
S^{yy}_t(T_0,T_n) = \frac{P^{y}_{tT_0}  - P^{y}_{tT_{n}}}{\sum_{i=1}^n \delta_i P^{y}_{tT_{i}}},
\end{equation}
for $t \leq T_0$. For a brief treatment of bootstrapping in an emerging market, we here refer to Appendix \ref{App-Boot-EM}.

\subsection{Discounting systems in developed markets}\label{dsdm}

Next we consider the more complex case of a \textit{developed market} where, in general, an OIS market exists. In such a market, cash flows are forecast using the $y$-tenored IBOR zero-coupon yield curve but discounted using the OIS zero-coupon yield curve. Such a product feature is also consistent with the notion of collateralisation. We consider the same FRA as in the \textit{emerging market} case, however we now assume that discounting occurs under the $x$-curve (or the OIS curve, to be more specific). We now have to make use of relation (\ref{xy-generic}) in order to define the FRA's payoff. 

\begin{Proposition}\label{xy-FRA-payoff-prop}
The developed market FRA with reset time $T_{i-1}$, expiry time $T_i$ and unit nominal has a terminal payoff, within the $x$-market, given by
\begin{equation}\label{xy-FRA-payoff}
V^{xy}_{T_iT_i} =  Q^{xy}_{T_{i-1}T_i} \delta_i \left(L^{y}_{T_{i-1}}(T_{i-1},T_i)-K^y \right) = Q^{xy}_{T_{i-1}T_i} V^{yy}_{T_iT_i}\,,
\end{equation}
where, as before, $\delta_i = T_i - T_{i-1}$ and $K^y$ is the strike rate within the $y$-market. The in-advance FRA payoff is then given by
\begin{equation}\label{xy-FRA-advance-payoff}
V^{xy}_{T_{i-1}T_i} = P^{x}_{T_{i-1}T_i} Q^{xy}_{T_{i-1}T_i} V^{yy}_{T_iT_i},
\end{equation}
which is the discounted value of the terminal payoff within the $x$-market.
\end{Proposition}

\begin{proof}
A direct application of relation (\ref{xy-generic}) leads to the result in Proposition \ref{xy-FRA-payoff-prop}. Like the \textit{emerging market} FRA, notice that the \textit{developed market} FRA's payoff is also measurable at $T_{i-1}$ with the actual cash flow occurring at $T_i$. 
\end{proof}

Before we consider the derivation of the value of the \textit{developed market} FRA, the following lemmas will prove to be useful in this regard.

\begin{Lemma}\label{xy-IBOR-mart-lemma}
The converted $y$-tenored forward IBOR process
\begin{equation}\label{xy-IBOR-mart}
L^{xy}_t(T_{i-1},T_i) = Q^{xy}_{tT_i}L^{y}_t(T_{i-1},T_i),
\end{equation}
for $t \in [0,T_{i-1}]$, satisfies the martingale relation 
\begin{equation}
L^{xy}_s(T_{i-1},T_i)=\frac{1}{h^{x}_s\,P^{x}_{sT_i}}\E\left[h^{x}_{T_i}L^{xy}_{T_{i-1}}(T_{i-1},T_{i})\,\big\vert\,\F_s\right],
\end{equation}
for $0 \leq s \leq t \leq T_{i-1}$.
\end{Lemma}
\begin{proof}
This statement follows from Eqs (\ref{xy-conv-factor}), (\ref{xy-IBOR-mart}) and (\ref{mart-IBOR}).
\end{proof}

\begin{Lemma}\label{xy-Forward-lemma}
The fair forward price $K^x$ of a forward contract initiated at time $t$ to exchange a cash flow $K^{y}$, determined in the $y$-market, for a cash flow of $K^{x}$, in the $x$-market, with $K^{y}$ being converted at $Q^{xy}_{T_{i-1}T_i}$ but the final payoff occurring at expiry $T_i > T_{i-1} \geq t$ is given by 
\begin{equation}
K^{x} =\frac{h^{y}_t P^{y}_{tT_i}}{h^{x}_t P^{x}_{tT_i}}K^{y}=Q^{xy}_{tT_i}K^{y}.
\end{equation} 
\end{Lemma}
\begin{proof}
The value of such a forward contract is given by
\begin{eqnarray}\label{xy-Forward}
V^{xy}_{tT_i} &=& \frac{1}{h^{x}_t}\E\left[h^{x}_{T_i}\left(K^{y} Q^{xy}_{T_{i-1}T_{i}} - K^{x} \right)\,\big\vert\,\F_t\right] \nn \\
	     &=& K^{y}\frac{h^{y}_t}{h^{x}_t} P^{y}_{tT_i} - P^{x}_{tT_i}K^{x},
\end{eqnarray}
which follows from Eq. (\ref{xy-conv-factor}) and the tower property of conditional expectations, while setting $V_{tT_i}^{xy} = 0$  and solving for $K^x$ yields the required result.
\end{proof}

We now have the necessary results to derive the value of the \textit{developed market} FRA, which is presented in the following theorem.

\begin{Theorem}\label{xy-IBOR-FRA-theorem}
The value of the developed market FRA with reset time $T_{i-1}$, expiry time $T_i$ and unit nominal, within the $x$-market, is given by
\begin{equation}\label{xy-IBORFRA}
V^{xy}_{tT_i} =  \delta_i P^{x}_{tT_i}  \left( L^{xy}_{t}(T_{i-1},T_i) - K^{x} \right),
\end{equation}
for $t \in [0,T_{i-1}]$, where $\delta_i = T_i - T_{i-1}$ and $K^x$ is the strike rate within the $x$-market.
\end{Theorem}
\begin{proof}
Using Proposition \ref{xy-FRA-payoff-prop}, the value of the \textit{developed market} FRA, for $t \in [0,T_{i-1}]$, is given by
\begin{eqnarray}
V^{xy}_{tT_i} 
	     &=& \frac{1}{h^{x}_t}\E\left[h^{x}_{T_{i-1}} P^{x}_{T_{i-1}T_i} Q^{xy}_{T_{i-1}T_i} \delta_i \left(L^{y}_{T_{i-1}}(T_{i-1},T_i)-K^y \right)  \,\big\vert\,\F_t\right] \nn \\
	     &=& \delta_i P^{x}_{tT_i} L^{xy}_{t}(T_{i-1},T_i) - \delta_i P^{x}_{tT_i} Q^{xy}_{tT_i}K^{y} ,
\end{eqnarray}
with the first term following from Lemma \ref{xy-IBOR-mart-lemma} and the second term from Eq. (\ref{xy-conv-factor}). Eq. (\ref{xy-IBORFRA}) follows from applying the result of Lemma \ref{xy-Forward-lemma} to the second term and factorising accordingly.
\end{proof}

The value of this FRA is commensurate with the value of a multi-curve (or basis) FRA in a {\it developed market}, i.e. the price dynamics are consistent with the standard FRA contract traded in developed markets. The form of the developed market FRA's value within the xy-framework leads to the following definition for the multi-curve forward IBOR process.

\begin{Definition}\label{DM-IBOR}
The multi-curve market-implied $y$-tenored forward IBOR process is given by
\begin{equation}\label{xy-IBOR}
L^{xy}_t(T_{i-1},T_i) = Q^{xy}_{tT_i}L^{y}_t(T_{i-1},T_i) = \frac{P^{xy}_{tT_i}}{\delta_i P^{x}_{tT_i}} \left( \frac{P^{xy}_{tT_{i-1}}}{P^{xy}_{tT_i}} - 1\right),
\end{equation}
for $t \in [0,T_{i-1}]$, where $(P^{xy}_{tT_i})$ is defined in Remark \ref{Rem-xyZCB}.
\end{Definition}

Moreover, we may also derive the fair \textit{developed market} FRA rate given the value of the FRA provided by Theorem \ref{xy-IBOR-FRA-theorem}.
\begin{Corollary}\label{fair-xy-rate}
The fair FRA rate process $K^{xy}_{t}(T_{i-1},T_i)$ at time $t$ of a developed market FRA written on the market-implied $y$-tenored forward IBOR (\ref{xy-IBOR-mart}), with reset time $T_{i-1}$ and settlement time $T_i$, is given by 
\begin{equation}\label{xy-FRA-rate}
K^{xy}_t(T_{i-1},T_i) = L^{xy}_t(T_{i-1},T_i),
\end{equation}
for $t \in [0,T_{i-1}]$.
\end{Corollary}
\begin{proof}
Setting the value of the developed market FRA, given by Eq. (\ref{xy-IBORFRA}), equal to zero, we find that $K^x = L^{xy}_t(T_{i-1},T_i)$ at time $t$. Then for any time $t \in [0,T_{i-1}]$, the result for the fair FRA rate process, $K^{xy}_t(T_{i-1},T_i) = L^{xy}_t(T_{i-1},T_i)$, follows accordingly.
\end{proof}

\begin{Remark}
Relation (\ref{xy-FRA-rate}) is the direct multi-curve analogy to the single-curve relation (\ref{yy-fair-FRA-rate}). In fact, for $x=y$ one recovers the single-curve expressions (\ref{yy-FRA}) and (\ref{yy-fair-FRA-rate}).
\end{Remark}

\begin{Remark}
Using Definition \ref{DM-IBOR}, one may re-state the value of the \textit{developed market} FRA as
\begin{equation}\label{xyquanto-FRA}
V^{xy}_{tT_i}= P^{xy}_{tT_{i-1}} - (1+ \delta_i K^y)P^{xy}_{tT_{i}},
\end{equation}
for $t \in [0,T_{i-1}]$, which is the direct multi-curve analogy to the \textit{emerging market} FRA value (\ref{yy-FRA}) with the $y$-ZCBs replaced by the $xy$-ZCBs.
\end{Remark}

Now that we have these results, it is also important to consider the relationship between $L^{xy}_{t}(T_{i-1},T_i)$ and $L^x_{t}(T_{i-1},T_i)$. In particular, one would want $L^{xy}_{t}(T_{i-1},T_i) \geq L^x_{t}(T_{i-1},T_i)$ due to the greater degree of risk associated with the multi-curve $y$-tenored forward IBOR process versus the corresponding $x$-tenored process. The following corollary reveals the conditions under which this feature is achieved, by making use of the  associated {\it forward capitalisation factor} (FCF) processes.

\begin{Corollary}\label{Corr-FCF}
The multi-curve market-implied $y$-tenored FCF process $\overline{v}^{xy}_t(T_{i-1},T_i)$, observed at time $t\le T_{i-1}$ and applying over the period $[T_{i-1},T_i]$, defined by
\begin{equation}\label{xy-FCF}
\overline{v}^{xy}_t(T_{i-1},T_i) := 1 + \delta_i L^{xy}_t(T_{i-1},T_i) \,,
\end{equation} 
is greater than or equal to the corresponding $x$-tenored FCF process 
\begin{equation}
v^x_t(T_{i-1},T_i) := 1 + \delta_i L^x_t(T_{i-1},T_i) \,, 
\end{equation} 
if interest rates are non-negative and $h^y_t \leq h^x_t$ for all $t \in [0,T_{i-1}]$ where $T_{i-1} \leq T_i$.
\end{Corollary}
\begin{proof}
Using Eq. (\ref{xy-FCF}) and Definition \ref{DM-IBOR}, we can show that
\begin{eqnarray}
\nonumber \overline{v}^{xy}_t(T_{i-1},T_i) &=& 1 + \delta_i Q^{xy}_{tT_i}  L^y_t(T_{i-1},T_i) \\
\nonumber &=& 1 + Q^{xy}_{tT_i}  \left(v^y_t(T_{i-1},T_i) - 1 \right) \\
\nonumber &=& 1 - Q^{xy}_{tT_i} + v^{xy}_t(T_{i-1},T_i) \,,
\end{eqnarray}
where $v^y_t(T_{i-1},T_i)$ is the $y$-tenored FCF and $v^{xy}_t(T_{i-1},T_i) := Q^{xy}_{tT_i} v^y_t(T_{i-1},T_i)$ is the $y$-tenored FCF represented equivalently in the $x$-market. Then, using Definition \ref{Qxy-Prop} and Eq. (\ref{xy-ZCB}), we can show that
\begin{equation*}
v^{xy}_t(T_{i-1},T_i) = Q^{xy}_{tT_i} \frac{P^y_{tT_{i-1}}}{P^y_{tT_i}} = Q^{xy}_{tT_i} \frac{P^{xy}_{tT_{i-1}}}{P^{xy}_{tT_i}} = \frac{P^x_{tT_{i-1}}}{P^x_{tT_i}}Q^{xy}_{tT_{i-1}} = v^{x}_t(T_{i-1},T_i) Q^{xy}_{tT_{i-1}} \,.
\end{equation*}
Now in order to have $v^{x}_t(T_{i-1},T_i) \leq \overline{v}^{xy}_t(T_{i-1},T_i)$, we must have that
\begin{align*}
v^{x}_t(T_{i-1},T_i) & \leq 1 - Q^{xy}_{tT_i} + v^{x}_t(T_{i-1},T_i)Q^{xy}_{tT_{i-1}} \\
v^{x}_t(T_{i-1},T_i) \left(1 - Q^{xy}_{tT_{i-1}}\right) & \leq 1 - Q^{xy}_{tT_i}  \\
1 - Q^{xy}_{tT_{i-1}} & \leq 1 - Q^{xy}_{tT_i}  \,,
\end{align*}
where the last inequality holds if interest rates are non-negative, i.e. $v^{x}_t(T_{i-1},T_i) \geq 1$. Finally, $Q^{xy}_{tT_{i-1}} \geq Q^{xy}_{tT_i}$ if interest rates are non-negative and $h^y_t \leq h^x_t$ for all $t \in [0,T_{i-1}]$ where $T_{i-1} \leq T_i$. This may be easily evidenced by setting $t=T_{i-1}$ and allowing $T_i$ to vary, while also using the linear and monotonic properties of conditional expectations.
\end{proof}
This corollary proves that the xy-approach, applied to a \textit{developed market}, yields a $y$-market interest rate system which is dominated by the $x$-market system, i.e. $P^y_{tT} \leq P^x_{tT}$ for $0 \leq t \leq T$. Furthermore, this $y$-market system provides a forward IBOR process, $L^y_t(T_{i-1},T_i)$, and enables the construction of a conversion factor process $Q^{xy}_{tT_i}$, which facilitates the definition of the developed market $y$-tenored forward IBOR process $L^{xy}_t(T_{i-1},T_i)$. Therefore, while the $y$-market system is still fictitious, given that it cannot be directly observed, we still consider it to be a model-consistent system given our curve-conversion framework that is inspired by currency modelling.
\begin{Remark}\label{alt-XY}
Using the FCF, one may also express the terminal payoff of the developed market FRA by
\begin{equation}\label{alt-xy-FRA-payoff}
V^{xy}_{T_iT_i} =  Q^{xy}_{T_{i-1}T_i} \left(v^{y}_{T_{i-1}}(T_{i-1},T_i) - v_K^y \right) \,,
\end{equation}
where $v_K^y := \left(1 + \delta_i K^y \right)$. Then, applying the same results as before, the value of the FRA, for $t \in [0,T_{i-1}]$, is
\begin{equation}\label{alt-xy-IBORFRA}
V^{xy}_{tT_i} = P^{x}_{tT_i}  \left( v^{xy}_{t}(T_{i-1},T_i) - v_K^{x} \right) \,,
\end{equation}
where $v_K^x = Q^{xy}_{tT_i} v_K^y$. If we define the multi-curve $y$-tenored forward IBOR process by
\begin{equation}\label{alt-xy-IBOR}
\overline{L}^{xy}_t(T_{i-1},T_i) := \frac{1}{\delta_i} \left(v^{xy}_{t}(T_{i-1},T_i)  - 1\right),
\end{equation}
and the multi-curve $x$-market equivalent FRA strike rate by
\begin{equation}\label{alt-FRA-strike}
\overline{K}^x := \frac{1}{\delta_i} \left(v_K^x  - 1\right),
\end{equation}
we then recover the developed market FRA price process: $V^{xy}_{tT_i} = P^x_{tT_i} \delta_i \left( \overline{L}^{xy}_t(T_{i-1},T_i) - \overline{K}^x \right)$.
We note that in this model $v^{xy}_t(T_{i-1},T_i) = v^{x}_t(T_{i-1},T_i) Q^{xy}_{tT_{i-1}}$ so that
\begin{equation}\label{Qxy-NS}
\frac{1 + \delta_i \overline{L}^{xy}_t(T_{i-1},T_i)}{1 + \delta_i L^x_t(T_{i-1},T_i)} =  \frac{h^y_t P^y_{tT_{i-1}}}{h^x_t P^x_{tT_{i-1}}}  \,,
\end{equation}
and $\overline{L}^{xy}_t(T_{i-1},T_i) \geq L^x_t(T_{i-1},T_i)$ if interest rates are non-negative and $h^x_t \leq h^y_t$ for all $t \in [0,T_{i-1}]$ and for $T_{i-1} \leq T_i$. This is the approach adopted by Nguyen \& Seifried \cite{ns} and it shall be revisited in Section \ref{R-Sec}. Two comments on their multi-curve model, given the context of the xy-approach, follow:
\begin{enumerate}
\item[(i)] The quantities $\overline{L}^{xy}_t(T_{i-1},T_i)$ and $\overline{K}^x$ which determine the FRA's floating and fixed cash flows are derived from the curve-converted quantities $v^{xy}_{t}(T_{i-1},T_i)$ and $v_K^x$ respectively. This is in contrast with $L^{xy}_t(T_{i-1},T_i)$ and $K^x$, the directly comparable curve-converted quantities used in the xy-framework. Therefore, these derived quantities are no longer consistent with a currency modelling analogy, with each differing from the correctly converted quantities by an additive factor of $(Q^{xy}_{tT_i} - 1)/\delta_i$.
\item[(ii)] Observation (i) is further supported by equation (\ref{Qxy-NS}) which shows that the conversion factor process effectively models the spread between the multi-curve $y$-tenored FCF and the corresponding $x$-tenored FCF, as opposed to the classical forward exchange rate. Moreover, the derived y-market system has almost no relation to the developed market y-tenored interest rate system, that one seeks to model, since the model derived y-market system dominates the x-market system, i.e.  $P^x_{tT} \leq P^y_{tT}$ for $0 \leq t \leq T$.
\end{enumerate}
\end{Remark}
\begin{Remark}
The mathematical quantity that directly models the $y$-tenored forward IBOR process is $L^{xy}_{\cdot}(\cdot,\cdot)$ and not $L^{y}_{\cdot}(\cdot,\cdot)$. This is a consequence of industry standards in developed markets, that the product of the $x$-pricing kernel and the $x$-curve discounted $y$-tenored forward IBOR process is a martingale under the $\PR$-measure. In the $xy$-approach, this implies that
\begin{equation}\label{DM-martingale-assumption}
h^{x}_s P^{x}_{sT_{i}} L^{y}_s(T_{i-1},T_i) = \E\left[h^{x}_t P^{x}_{tT_{i}} L^{y}_t(T_{i-1},T_i) \,\vert\,\F_s \right],
\end{equation}
for $0 \leq s \leq t \leq T_{i-1}$. It is not possible to achieve this relationship within the xy-framework, given our representation of the $y$-tenored forward IBOR process (\ref{y-IBOR}). However this relationship is achieved if we replace $L^{y}_{\cdot}(\cdot,\cdot)$ with $L^{xy}_{\cdot}(\cdot,\cdot)$. Our market-implied $y$-tenored forward IBOR process, $L^{xy}_{\cdot}(\cdot,\cdot)$, reveals the convolution of a conversion factor (which facilitates the market's martingale assumption (\ref{DM-martingale-assumption})) and the model $y$-tenored forward IBOR process, $L^{y}_{\cdot}(\cdot,\cdot)$. This result questions the utility of the $y$-ZCB system in the \textit{developed market} context. The $y$-ZCB system is a model construct, derived from the $y$-tenored model-consistent or model-implied forward IBOR process, $L^{y}_{\cdot}(\cdot,\cdot)$, which unravels the market's martingale adjustment from the observed $y$-tenored market-implied IBOR process, $L^{xy}_{\cdot}(\cdot,\cdot)$, via the conversion factor $Q^{yx}_{\cdot \cdot}$.
\end{Remark}
\begin{Remark}
The xy-framework advocates the following price process for a multi-curve FRA
\begin{equation}\label{model-xy-IBORFRA}
V^{xy}_{tT_i} = \delta_i P^{xy}_{tT_i} \left( L^{y}_t(T_{i-1},T_i) - K^{y} \right) ,
\end{equation}
for $t \in [0,T_i]$. We note that the conversion factor (or martingale adjustment) has been applied to the discounting x-ZCB system and not to the model for the $y$-tenored forward IBOR process. However, we note that the terminal FRA payoff would now be
\begin{equation*}
V^{xy}_{T_iT_i} = \delta_i \frac{h^y_{T_i}}{h^x_{T_i}} \left( L^{y}_{T_{i-1}}(T_{i-1},T_i) - K^{y} \right) .
\end{equation*}
This allows us to \textit{disentangle} the y-ZCB system from the x-ZCB system, which enables us to model the y-curve discounting in a consistent, robust and rigorous fashion. From an economics perspective, if one compares the return generated from an xy-FRA to a yy-FRA, one can show that 
\begin{equation}
\frac{V^{yy}_{tT_i}}{V^{yy}_{0T_i}} > \frac{V^{xy}_{tT_i}}{V^{xy}_{0T_i}} = \frac{V^{xy}_{tT_i}}{V^{yy}_{0T_i}} ,
\end{equation}
as required, since discounting at the $x$-curve essentially represents a collateralised FRA which should therefore return the holder less than an equivalent investment in a non-collateralised FRA, represented by the $y$-curve discounting.
\end{Remark}
Next we consider the \textit{developed market} IRS, i.e. one which forecasts cash flows under the $y$-curve but discounts under the $x$-curve, unlike the \textit{emerging market} IRS.
\begin{Theorem}\label{xy-IRS-theorem}
The value of a developed market IRS, within the $x$-market, with reset times $\{T_0,T_1,\ldots,T_{n-1}\}$, payment times $\{T_1,T_2,\ldots,T_n\}$ and unit nominal, referencing the $y$-tenored IBOR is given by
\begin{equation}\label{xy-IRS} 
V^{xy}_{tT_n} = \sum_{i=1}^n \delta_i P^{x}_{tT_i}   \left(L^{xy}_{t}(T_{i-1},T_i) - S^x \right),
\end{equation}
for $t \leq T_0$, where $\delta_i = T_i - T_{i-1}$ and where $S^x$ is the fixed swap rate within the $x$-market.
\end{Theorem}
\begin{proof}
Starting with the \textit{emerging market} version of the IRS with fixed swap rate $S^y$ within the $y$-market and applying pricing relation (\ref{xy-generic}), analagous to Proposition \ref{xy-FRA-payoff-prop}, the \textit{developed market} IRS price process is given by
\begin{equation}
V^{xy}_{tT_n} = \sum_{i=1}^n \frac{\delta_i}{h^{x}_t}\E\left[h^{x}_{T_{i-1}} P^{x}_{T_{i-1}T_i} Q^{xy}_{T_{i-1}T_i} \left(L^{y}_{T_{i-1}}(T_{i-1},T_i) - S^y \right) \,\big\vert\,\F_t\right] ,
\end{equation} 
which, upon application of Lemma \ref{xy-IBOR-mart-lemma} and Equation (\ref{xy-conv-factor}), simplifies to 
\begin{equation}
V^{xy}_{tT_n} = \sum_{i=1}^n \delta_i P^{x}_{tT_i}\left( L^{xy}_{t}(T_{i-1},T_i) - Q^{xy}_{tT_i}S^y \right) ,
\end{equation} 
for $t \leq T_0$. The result follows by observing that the fixed IRS rate may be expressed in the $x$-market by $S^x = S^y(\sum_{i=1}^n \delta_i P^{x}_{tT_i} Q^{xy}_{tT_i})/(\sum_{i=1}^n \delta_i P^{x}_{tT_i})$. This may be justified in an analogous fashion to the fixed FRA rate, but this time making use of a fixed-for-fixed swap contract as opposed to a forward contract, as in Lemma \ref{xy-Forward-lemma}.
\end{proof}

\begin{Remark}
Using Definition \ref{xy-IBOR}, one may re-state the value of the \textit{developed market} IRS as
\begin{equation}\label{xy-IRS2}
V^{xy}_{tT_n}= \left[P^{xy}_{tT_0}  - P^{xy}_{tT_{n}} \right] - S^y \sum_{i=1}^n \delta_i P^{xy}_{tT_{i}} ,
\end{equation}
for $t \leq T_0$, which is the direct multi-curve analogy to the \textit{emerging market} IRS value (\ref{yy-IRS}) with the $y$-ZCBs replaced by the $xy$-ZCBs.
\end{Remark}

\begin{Corollary}\label{fair-xy-IRS-rate}
The fair fixed swap rate process $S^{xy}_t(T_0,T_n)$ of a developed market IRS written on the market-implied $y$-tenored forward IBOR (\ref{xy-IBOR-mart}), with reset times $\{T_0,T_1,\ldots,T_{n-1}\}$, payment times $\{T_1,T_2,\ldots,T_{n}\}$ and unit nominal, is given by 
\begin{equation}\label{xy-fair-IRS-rate-process}
S^{xy}_t(T_0,T_n) = \frac{P^{xy}_{tT_0}  - P^{xy}_{tT_{n}}}{\sum_{i=1}^n \delta_i  P^{x}_{tT_{i}}} ,
\end{equation}
for $t \leq T_0$.
\end{Corollary}
\begin{proof}
Setting the value of the developed market IRS equal to zero, given by Eq. (\ref{xy-IRS2}), it follows that the $y$-market fair fixed IRS rate is $S^y = \left(P^{xy}_{tT_0}  - P^{xy}_{tT_{n}}\right)/\sum_{i=1}^n \delta_i  P^{xy}_{tT_{i}}$ at time $t$. Using the proof of Theorem \ref{xy-IRS-theorem} and Remark \ref{Rem-xyZCB}, the $x$-market fair fixed IRS rate (converting the $y$-market rate) is given by $S^x = S^y(\sum_{i=1}^n \delta_i P^{x}_{tT_i} Q^{xy}_{tT_i})/(\sum_{i=1}^n \delta_i P^{x}_{tT_i})$ at time $t$. Then for any time $t \leq T_0$, the result for the developed market fair IRS rate follows accordingly by setting $S^{xy}_t(T_0,T_n) = S^x$.
\end{proof}

For a brief treatment of bootstrapping in a developed market, we here refer to Appendix \ref{App-Boot-DM}.

\subsection{Consistent multi-curve discounting in emerging markets}

Now that we have a good understanding of how the xy-formalism enables the modelling of multi-curve interest rate systems in  \textit{developed markets}, we may consider resolving the same problem for the case of an \textit{emerging market}. Our first hurdle in moving from a \textit{developed} to an \textit{emerging market} setting is the non-existence of the OIS curve. 

Recall that we have assumed the existence of a collection of interest rate curves indexed by $x, y = 0,1,2, \ldots,n$ where we refer to the $x$-curve as the \textit{discounting curve} and the $y$-curve as the \textit{forecasting curve}. In a common \textit{developed market}, $n=4$ with $0$ denoting the nominal OIS curve, $1$ the 1-month IBOR curve, $2$ the 3-month IBOR curve, $3$ the 6-month IBOR curve and $4$ the 12-month IBOR curve. Moreover, the stochastic evolution of each of these curves are modelled via a pricing kernel process $(h_t^y)$ which are in turn calibrated using liquid linear and non-linear interest rate market instruments. In a common \textit{emerging market}, only one IBOR tenor is usually tradable and liquid, therefore it is not possible to calibrate the entire set of pricing kernel processes $(h_t^y)$ which span the common developed interest rate market. This leads us to the following remark.
\begin{Remark}\label{Est-EM-PKs-prop} 
In the common emerging market, only one IBOR tenor, $y^*$, is tradable and liquid thereby enabling the specification and calibration of a well-defined pricing kernel process $(h_t^{y^*})$. Pricing kernel processes for all other IBOR tenors $(\widehat{h}_t^{y})$ are to be estimated statistically (or otherwise) as a suitable functional form of  $(h_t^{y^*})$, i.e.
\begin{equation}\label{Est-EM-PK}
\widehat{h}_t^{y} = f\left( h_t^{y^*} \right) .
\end{equation}
where $f:\mathbb{R}^+ \rightarrow \mathbb{R}^+$ is measurable and adapted, such that the corresponding estimated $y$-ZCB (and $y$-curve) systems, $(\widehat{P}_{tT}^{y})$, may be constructed via
\begin{equation}
\widehat{P}_{tT}^{y} = \frac{1}{\widehat{h}_t^{y}} \E\left[\widehat{h}^{y}_T\,\vert\,\F_t\right] = \frac{1}{f\left( h_t^{y^*} \right)} \E\left[f\left( h_T^{y^*} \right)\,\vert\,\F_t\right] ,
\end{equation}
for $0 \leq t \leq T$.
\end{Remark}
In Remark \ref{Est-EM-PKs-prop}, if the function $f(\cdot)$ is linear, then the estimated $y$-ZCB is given by
\begin{equation}
\widehat{P}_{tT}^{y} = \frac{1}{f\left( h_t^{y^*} \right)} f\left( P_{tT}^{y^*} \right) ,
\end{equation}
which implies that it is possible to directly replicate the estimated $y$-ZCB through either a static or dynamic replication strategy using the $y^*$-ZCB. However, this may not be possible, in general, if the function $f(\cdot)$ is convex (concave), as the estimated $y$-ZCB will be governed by the following inequality
\begin{equation}
\widehat{P}_{tT}^{y} \geq (\leq) \frac{1}{f\left( h_t^{y^*} \right)} f\left( P_{tT}^{y^*} \right) ,
\end{equation}
which follows by the application of Jensen's inequality. The xy-formalism may now be applied in the \textit{emerging market} setting, assuming the existence of a collection of interest rate curves, indexed by $x, y = 0,1,\ldots, y^*,\ldots, n$, that are modelled by the calibrated pricing kernel process $(h_t^{y^*})$ and the set of estimated pricing kernel processes $(\widehat{h}_t^y;y \neq y^*)$. First, we consider the \textit{developed market} FRA within the \textit{emerging market} context, i.e. one where the payoff is forecasted by the $y$-curve and then discounted by the $x$-curve. The terminal and in-advance FRA payoffs remain unchanged and are identical to Eqs (\ref{xy-FRA-payoff}) and (\ref{xy-FRA-advance-payoff}), respectively\footnote{Assuming that one insists on maintaining measurability of the payoff at the IBOR reset time $T_{i-1}$.}, with the FRA price process also assuming the familiar form
\begin{eqnarray}\label{xy-EM-IBORFRA}
 V^{xy}_{tT_i} &=& \delta_i P^{xy}_{tT_i} \left(L^{y}_{t}(T_{i-1},T_i)-K^{y} \right) \nn \\
	     &=& P^{xy}_{tT_{i-1}} - (1+ \delta_i K^y)P^{xy}_{tT_{i}},
\end{eqnarray}
for $t \in [0,T_{i-1}]$, while $L^{y}_{t}(T_{i-1},T_i)$ continues to be the correct forward IBOR process. Notice that the derivation of equation (\ref{xy-EM-IBORFRA}) follows by a direct application of Corollary \ref{Qxy-dual-coro}.

\begin{Definition}\label{EM-IBOR}
The multi-curve emerging market $y$-tenored forward IBOR process is given by
\begin{equation}\label{xy-EM-IBOR-mart}
L^{y}_t(T_{i-1},T_i) =  \frac{1}{\delta_i} \left( \frac{P^{y}_{tT_{i-1}}}{P^{y}_{tT_i}} - 1\right),
\end{equation}
for $t \in [0,T_{i-1}]$, unlike the developed market which required the definition of the market-implied $y$-tenored forward IBOR process $L^{xy}_{t}(T_{i-1},T_i) := Q^{xy}_{tT_i} L^{y}_{t}(T_{i-1},T_i)$ for $t \in [0,T_{i-1}]$.
\end{Definition}
This is due to the fact that there is currently no market standard for pricing an \textit{emerging market} FRA that is forecasted and discounted under different curves, with the only observable market quantity being the spot IBOR process $L^{y}_{t}(t,t+\delta)$ for $t \geq 0$. It is also possible, as in the case of \textit{developed markets}, to define a fair FRA rate process, $K_t^{xy}(T_{i-1},T_i) = L_t^{y}(T_{i-1},T_i)$, however one would not be able to observe this quantity in the market (since these FRAs are not traded, in general), therefore this would be a model-implied quantity\footnote{If the $y$-tenored IBOR corresponds to the most liquid and tradable tenor, i.e. $y=y^*$, then one will also have access to the set of forward IBOR processes $L^{y}_{t}(T_{i-1},T_i)$ for $0 \leq t \leq T_{i-1}$, from the standard and liquidly tradable set of single-curve \textit{emerging market} FRAs, and $K^{xy}_t(T_{i-1},T_i) = K^{yy}_t(T_{i-1},T_i) = L^{y}_t(T_{i-1},T_i)$.}.
\

Similarly, we may consider the standard \textit{developed market} IRS in the context of an \textit{emerging market}. The value of the IRS at some time $t \leq T_0$, making use of the same relations as before, is again given by
\begin{eqnarray} \label{xy-EM-IRS}
V^{xy}_{tT_n}
&=& \sum_{i=1}^n \delta_i P^{xy}_{tT_i} \left(L^{y}_{t}(T_{i-1},T_i) - S^y \right) \nn\\
&=&P^{xy}_{tT_0}  - P^{xy}_{tT_{n}} - S^y \sum_{i=1}^n \delta_i P^{xy}_{tT_{i}},
\end{eqnarray} 
where $L^{y}_{t}(T_{i-1},T_i)$, for $t \in [0,T_{i-1}]$, continues to be the correct forward IBOR process, analagous to the FRA result. As with the FRA, the fair IRS rate process is model-implied (unless the $y$-tenored IBOR process is the tradable tenor and $t=T_0=0$) and given by $S_t^{xy}(T_0,T_n)=(P^{xy}_{tT_0}  - P^{xy}_{tT_{n}})/(\sum_{i=1}^n \delta_i  P^{xy}_{tT_{i}})$ for $t \leq T_0$.
\

In a multi-curve \textit{emerging market} interest rate system, within the xy-framework, the initial (estimated) $y$-ZCB systems may be constructed in a completely analogous fashion to the single-curve \textit{emerging market} relations, see Appendix \ref{App-C}, since $K_0^{xy}(T_{i-1},T_i) = L_0^{y}(T_{i-1},T_i)$ and $S_0^{xy}(0,T_n)=(1  - P^{y}_{0T_{n}})/(\sum_{i=1}^n \delta_i  P^{y}_{0T_{i}})$. That is, all initial model-implied quantities are only dependent on the $y$-curve or $y$-ZCB system.
\

If we consider a FRA and an IRS within this context with payoffs forecasted by the $y^*$-curve and discounted by one of the other curves, denoted by the $x$-curve, then the pricing formulae are given by
\begin{equation}\label{xy*-EM-IBORFRA}
 V^{xy^*}_{tT_i} = P^{xy^*}_{tT_{i-1}} - (1+ \delta_i K^{y^*})P^{xy^*}_{tT_{i}},
\end{equation}
and
\begin{equation} \label{xy*-EM-IRS}
V^{xy^*}_{tT_n} = \left[P^{xy^*}_{tT_0}  - P^{xy^*}_{tT_{n}} \right] - S^y \sum_{i=1}^n \delta_i P^{xy^*}_{tT_{i}},
\end{equation} 
from Eqs (\ref{xy-EM-IBORFRA}) and (\ref{xy-EM-IRS}), respectively. At this juncture, it is important to note that the $xy^*$-ZCB, $(P_{tT}^{xy^*})$, plays the same role as the $y^*$-ZCB, $(P_{tT}^{y^*})$, does in the single-curve emerging market setting in Section \ref{DS-EM}. This leads us to the following definition for the $xy$-ZCB system, in general.

\begin{Definition}\label{quanto-bond}
In the multi-curve interest rate system derived within the xy-framework, the $xy$-ZCB system, $(P_{tT}^{xy})$, defined by
\begin{equation*}
P^{xy}_{tT}=\frac{1}{h^{x}_t}\E\left[h^{x}_T \left(\frac{h^{y}_T}{h^{x}_T}\right)(1) \,\big\vert\,\F_t\right] = P_{tT}^{x}Q_{tT}^{xy} = Q_{tt}^{xy} P_{tT}^{y},
\end{equation*}
may be considered to be a quanto-bond assuming
\begin{enumerate}
	\item[(i)] the $x$-curve with varying notional defined by the forward conversion factor $Q_{tT}^{xy}$; or 
	\item[(ii)] the $y$-curve with varying notional defined by the spot conversion factor $Q_{tt}^{xy}$.
\end{enumerate}
\end{Definition}

\begin{Remark}
Within the developed market context---where the nominal OIS curve is considered to be the distinct, single-curve tradable system, which we shall denote here as the $x^*$-curve---one may dynamically replicate $y$-ZCBs and $x^*y$-ZCBs, where $y \neq x^*$, via the following set of $x^*$-curve quanto-bonds
\begin{equation*}
P_{tT}^{y} = \frac{Q_{tT}^{x^*y}}{Q_{tt}^{x^*y}}P_{tT}^{x^*} \;\; \mathrm{and} \;\; P_{tT}^{x^*y} = Q_{tT}^{x^*y}P_{tT}^{x^*} ,
\end{equation*}
whereas, within the emerging market context where one nominal IBOR swap curve is considered to be the distinct, single-curve tradable system, which we have denoted as the $y^*$-curve, one may dynamically replicate $x$-ZCBs and $xy^*$-ZCBs, where $x \neq y^*$, via the following set of $y^*$-curve quanto-bonds
\begin{equation*}
P_{tT}^{x} = \frac{Q_{tt}^{xy^*}}{Q_{tT}^{xy^*}}P_{tT}^{y^*} \;\; \mathrm{and} \;\; P_{tT}^{xy^*} = Q_{tt}^{xy^*}P_{tT}^{y^*} .
\end{equation*}
\end{Remark}

\section{xy-HJM multi-curve models}\label{xy-HJM-Sec}

In this section we develop Heath-Jarrow-Morton (HJM) multi-curve interest rate systems based on the xy-formalism introduced in this paper. The xy-HJM multi-curve system will be derived using results from Section \ref{PK-Multi}.
\

We consider the filtered probability space $(\Omega, \F, (\F_t), \PR)$ where $(\F_t)_{0\le t}$ is the filtration generated by two sets of independent multi-dimensional $\PR$-Brownian motions $(W_t)_{t \ge 0}$ and $(Z_t)_{t \ge 0}$, respectively. Being synonymous with the xy-formalism, we consider an economy with two distinct markets, $x$ and $y$, where $x$ may be interpreted as a proxy default-free OIS-based market and $y$ as a risky IBOR-based market. Furthermore, we assume that the $x$- and $y$-markets are driven by the multi-dimensional $\PR$-Brownian motions $(W_t^x)_{t \ge 0} = (W_t)_{t \ge 0}$ and $(W_t^y)_{t \ge 0} = (W_t, Z_t)_{t \ge 0}$ respectively, where $(W_t)_{t \ge 0}$ is $n$-dimensional and $(Z_t)_{t \ge 0}$ is $m$-dimensional. This allows us to define the pricing kernel process associated with each market.
\begin{Definition}
The $(\F_t)$-adapted $x$- and $y$-market pricing kernel processes $(h^x_t)_{0\le t}$ and $(h^y_t)_{0\le t}$ satisfy, respectively,
\begin{align}\label{xy-HJM-PK}
&\frac{\rd h^x_t}{h^x_t}=-r^x_t\rd t - \lambda^x_t\rd W^x_t,& &\frac{\rd h^y_t}{h^y_t}=-r^y_t\rd t - \lambda^y_t\rd W^y_t,&
\end{align}
where $(r^x_t)_{t \ge 0}$ and $(r^y_t)_{t \ge 0}$ are the short rates of interest; and $(\lambda^x_t)_{t \ge 0}$ and $(\lambda^y_t)_{t \ge 0} = (\lambda^x_t, \lambda^z_t)_{t \ge 0}$ are the $n$- and $(n+m)$-dimensional market price of risk processes associated with the $x$- and $y$-markets, respectively.
\end{Definition}
Next, let $(X_{tT})_{0\le t\le T}$ and $(Y_{tT})_{0\le t\le T}$ be (well-defined) processes, respectively satisfying
\begin{align}\label{xy-HJM-marts}
\nonumber &\frac{\rd X_{tT}}{X_{tT}}= \left(-A_{tT}^x + \frac{1}{2} \left| \Sigma^x_{tT} \right|^2 \right) \rd t - \left(\Sigma^x_{tT} + \lambda^x_t \right) \rd W^x_t,\\
&\frac{\rd Y_{tT}}{Y_{tT}}= \left(-A_{tT}^y + \frac{1}{2} \left| \Sigma^y_{tT} \right|^2 \right) \rd t - \left(\Sigma^y_{tT} + \lambda^y_t \right) \rd W^y_t,&
\end{align}
for $0 \le t \le T$, where $A^{(\cdot)}_{tT} = \int_t^T a^{(\cdot)}_{tu} \rd u$ is $1$-dimensional, $| \; \cdot \; |$ denotes the Euclidean norm, and $\Sigma^{(\cdot)}_{tT} = \int_t^T \sigma^{(\cdot)}_{tu} \, \rd u$ with $\Sigma^{x}_{tT}$ and $\Sigma^{y}_{tT}$ being $n$- and $(n+m)$-dimensional, respectively. The processes $(\sigma_{tT}^{x})$, $(\sigma_{tT}^{y}) = (\sigma_{tT}^{w}, \sigma_{tT}^{z})$ and $(a^{(\cdot)}_{tT})$ are generic adapted processes satisfying the implicit integrability conditions, with $\sigma_{tT}^w$ and $\sigma_{tT}^z$ being $n$- and $m$-dimensional, respectively. We may then define the respective ZCB prices as follows:

\begin{Definition}\label{xy-HJM-ZCB-defn}
Setting $P^{x}_{tT} := X_{tT} / h^{x}_t$ and $P^{y}_{tT} := Y_{tT} / h^{y}_t$, the $(\F_t)$-adapted $x$- and $y$-market ZCB-systems satisfy, respectively, the dynamical equations
\begin{eqnarray}\label{xy-HJM-ZCB}
\nonumber \frac{\rd P^x_{tT}}{P^x_{tT}}&=&\left(r^x_t -A_{tT}^x + \frac{1}{2} \left| \Sigma^x_{tT} \right|^2 - \lambda^x_t \Sigma^x_{tT} \right)\rd t - \Sigma^x_{tT}\rd W^x_t,\\
\frac{\rd P^y_{tT}}{P^y_{tT}}&=&\left(r^y_t -A_{tT}^y + \frac{1}{2} \left| \Sigma^y_{tT} \right|^2 - \lambda^y_t \Sigma^y_{tT} \right)\rd t - \Sigma^y_{tT}\rd W^y_t,
\end{eqnarray}
following the application of Ito's Lemma. Invoking the classical HJM drift condition, $A_{tT}^{(\cdot)} = \frac{1}{2} |\Sigma^{(\cdot)}_{tT}|^2$, results in $(h^{(\cdot)}_t P^{(\cdot)}_{tT})_{0 \le t \le T}$ being a $\PR$-(local) martingale which is a requirement for the xy-HJM framework. 
\end{Definition}
\begin{Proposition}
Assuming that the $x$- and $y$-market ZCB-systems are differentiable in $T$, the instantaneous forward rate processes $(f^x_{tT})_{0\le t\le T}$ and $(f^{y}_{tT})_{0\le t\le T}$, respectively defined by $f^{x}_{tT}=-\partial_T\ln\left(P^{x}_{tT}\right)$ and $f^{y}_{tT}=-\partial_T\ln\left(P^{y}_{tT}\right)$, satisfy
\begin{eqnarray}
\nonumber \rd f^x_{tT}&=&\left(a^x_{tT} + \lambda^x_t \sigma^x_{tT} \right)\rd t + \sigma^x_{tT}\rd W^x_t,\\
\rd f^y_{tT}&=&\left(a^y_{tT} + \lambda^y_t \sigma^y_{tT} \right)\rd t + \sigma^y_{tT}\rd W^y_t,
\end{eqnarray}
which are consistent with the classical HJM instantaneous forward rate model.
\end{Proposition}
\begin{proof}
By direct application of Ito's Lemma, the logarithm of the ZCB price process is
\begin{equation}
\ln\left(P^{(\cdot)}_{tT} \right) = \ln\left(P^{(\cdot)}_{0T} \right) + \int_0^t \left(r_s - A^{(\cdot)}_{sT} - \lambda^{(\cdot)}_s \Sigma^{(\cdot)}_{sT} \right) \rd s - \int_0^t \Sigma^{(\cdot)}_{sT} \rd W_s^{(\cdot)}, 
\end{equation}
and therefore taking the negative and differentiating with respect to $T$ gives
\begin{equation}
-\frac{\partial }{\partial T}\ln\left(P^{(\cdot)}_{tT} \right) = \int_0^t \left( a^{(\cdot)}_{sT} + \lambda^{(\cdot)}_s \sigma^{(\cdot)}_{sT} \right) \rd s + \int_0^t \sigma^{(\cdot)}_{sT} \rd W_s^{(\cdot)}, 
\end{equation}
which yields the required instantaneous forward rate result.
\end{proof}

Grbac \& Runggaldier \cite{gr} provide a thorough account of the approaches that have been adopted in modeling a developed market multi-curve interest rate system with the HJM framework. Here we reprise the key results, given the economy that has already been introduced in this section, in order to contextualise the xy-HJM framework within the existing body of literature. Grbac \& Runggaldier \cite{gr} note that all approaches that have been adopted model the $x$-market ZCB-system $(P^x_{tT})$ with the classical HJM model while the multi-curve market-implied $y$-tenored forward IBOR process, which we denote here by $\overline{L}^{xy}_t(T_{i-1},T_i)$, is modeled in one of three ways:
\begin{enumerate}
	\item[(i)] $\overline{L}^{xy}_t(T_{i-1},T_i)$ is specified in an ad hoc fashion (usually) inspired by the LIBOR Market Model (LMM) such that this approach is referred to as a \textit{hybrid HJM-LMM};
	\item[(ii)] $\overline{L}^{xy}_t(T_{i-1},T_i) :=( v^y_t(T_{i-1},T_i) - 1)/\delta_i$ where, as before, $v^y_t(T_{i-1},T_i) := P^y_{tT_{i-1}}/P^y_{tT_{i}}$ defines the FCF such that under certain parameter restrictions (see Proposition \ref{modelii} below) $(h_t^x P^x_{tT_i} v^y_t(T_{i-1},T_i))_{0 \le t \le T_{i-1}}$ is a $\PR$-(local) martingale; and
	\item [(iii)] $\overline{L}^{xy}_t(T_{i-1},T_i) :=\E[ h^x_{T_i} L^y_{T_{i-1}}(T_{i-1},T_{i}) \,\big\vert\,\F_t]/(h^x_t P^x_{tT_i})$ assuming the classical HJM drift condition for the $y$-market ZCB system.
\end{enumerate}
In each approach $(h_t^x P^x_{tT_i} \overline{L}^{xy}_t(T_{i-1},T_i))_{0 \le t \le T_{i-1}}$ is a $\PR$-(local) martingale, as required. Model (i) is inconsistent with our approach, since our focus is on modeling ZCB-systems directly and implying simple spot and forward rate models therefrom, therefore we merely make note of (i) for completeness. Models (ii) and (iii) are comparable to our approach, therefore we expand upon them below.
\begin{Proposition}\label{modelii}
If the following parameter restrictions hold:
\begin{equation}\label{modelii-condition}
A^y_{tT_i} - A^y_{tT_{i-1}} = -\frac{1}{2}\left|\Sigma^y_{tT_i} - \Sigma^y_{tT_{i-1}} \right|^2 + \Sigma^x_{tT_i}\left(\Sigma^w_{tT_i}-\Sigma^w_{tT_{i-1}} \right) -\lambda^z_t \left(\Sigma^z_{tT_i} - \Sigma^z_{tT_{i-1}}\right),
\end{equation}
then the process $(h^x_t P^x_{tT_i} v^y_t(T_{i-1},T_i))_{0 \le t \le T_{i-1}}$ a $\PR$-(local) martingale, thereby enabling the use of Model (ii).
\end{Proposition}
\begin{proof}
Applying Ito's Lemma to $h^x_t P^x_{tT_i} v^y_t(T_{i-1},T_i)$, using Eqs (\ref{xy-HJM-PK}) and (\ref{xy-HJM-ZCB}), we have
\begin{eqnarray}
\nonumber & & \frac{\rd \left(h^x_t P^x_{tT_i} v^y_t(T_{i-1},T_i)\right)}{h^x_t P^x_{tT_i} v^y_t(T_{i-1},T_i)} \\
\nonumber &=& \left[A^y_{tT_i} - A^y_{tT_{i-1}} +\frac{1}{2}\left|\Sigma^y_{tT_i} - \Sigma^y_{tT_{i-1}} \right|^2 - \Sigma^x_{tT_i}\left(\Sigma^w_{tT_i}-\Sigma^w_{tT_{i-1}} \right) +\lambda^z_t \left(\Sigma^z_{tT_i} - \Sigma^z_{tT_{i-1}}\right) \right]\rd t  \\
	& & + \left(\Sigma^w_{tT_i} - \Sigma^w_{tT_{i-1}} - \Sigma^x_{tT_i} - \lambda^x_t  \right) \rd W^x_t + \left(\Sigma^z_{tT_i} - \Sigma^z_{tT_{i-1}}\right) \rd Z_t,
\end{eqnarray}
from which it follows that the required martingale condition is achieved only if Eq. (\ref{modelii-condition}) is enforced.
\end{proof}

\begin{Remark}
In Grbac \& Runggaldier \cite{gr}, both the $x$- and $y$-markets have the same sources of risk, i.e. are driven by the same set of Brownian motions, which resolves the parameter restrictions to
\begin{equation}
A^y_{tT_i} - A^y_{tT_{i-1}} = -\frac{1}{2}\left(\Sigma^y_{tT_i} - \Sigma^y_{tT_{i-1}} \right)^2 + \Sigma^x_{tT_i}\left(\Sigma^y_{tT_i}-\Sigma^y_{tT_{i-1}} \right),
\end{equation}
for $0 \le t \le T_{i-1}$. 
\end{Remark}
Model (iii) requires one to compute a conditional expectation, $\E[ h^x_{T_i} L^y_{T_{i-1}}(T_{i-1},T_{i}) \,\big\vert\,\F_t]$, which is possible given our model choices, i.e. Eqs (\ref{xy-HJM-PK}) and (\ref{xy-HJM-ZCB}), along with the classical HJM drift condition applied to the $y$-market ZCB-system. Note that in Grbac \& Runggaldier \cite{gr}, this model is justified by analogies to credit and foreign exchange modeling. Their model setup leads to two different parameter restrictions, depending on which analogy is assumed. Our pricing kernel-based HJM setup leads to a unique parameter restriction (the classical HJM drift condition for the $y$-market ZCB-system) which subsumes both analogies, since our setup does not require us to specify an exchange rate process in an ad hoc exogenous manner. This may be seen in the following proposition, recalling the results from Proposition \ref{Qxy-Prop}.

\begin{Proposition}
The xy-HJM framework's forward curve-conversion factor process $(Q^{xy}_{tT})_{0\le t\le T}$ satisfies
\begin{equation}
\frac{\rd Q^{xy}_{tT}}{Q^{xy}_{tT}}= (\Sigma^x_{tT} - \Sigma^w_{tT})  (\Sigma^x_{tT} \rd t + \lambda^x_t \rd t + \rd W^x_t) - (\Sigma^z_{tT} + \lambda^z_t) \rd Z_t ,
\end{equation}
while the spot curve-conversion factor process $(Q^{xy}_{tt})_{t\ge 0}$ satisfies
\begin{equation} 
\frac{\rd Q^{xy}_{tt}}{Q^{xy}_{tt}}=\left(r^x_t-r^y_t\right)\rd t  + \lambda^x_t \rd W^x_t -\lambda^y_t \rd W^y_t,
\end{equation}
where $\lambda^x_t \rd W^x_t -\lambda^y_t \rd W^y_t = -\lambda_t^z \rd Z_t$.
\end{Proposition}
\begin{proof} 
Using the definition of the conversion factor, Eq. (\ref{xy-conv-factor}), along with Definition \ref{xy-HJM-ZCB-defn}, observe that $Q^{xy}_{tT} = Y_{tT}/X_{tT}$ while $Q^{xy}_{tt} = h^y_t/h^x_t$. The result then follows by a straightforward application of Ito's Lemma using Eqs. (\ref{xy-HJM-marts}) and (\ref{xy-HJM-PK}).
\end{proof}

\begin{Remark}
By the Girsanov Theorem, it is straightforward to show that there is a multi-dimensional $\Q_x$-Brownian motion $(W^{\Q_x}_t)$ that satisfies $\rd W^{\Q_x}_t=\lambda^x_t\rd t + \rd W^x_t$ upon changing measure from $\PR$ to the $x$-market risk-neutral measure $\Q_x$. Moreover there is also a multi-dimensional $\Q_x^T$-Brownian motion $(W^{\Q_x^T}_t)$ that satisfies $\rd W^{\Q_x^T}_t = \Sigma^x_{tT} \rd t + \rd W^{\Q_x}_t$ upon changing measure from $\Q_x$ to the $x$-market $T$-forward measure $\Q_x^T$.
\end{Remark}

In this paper we have proposed the xy-formalism for multi-curve interest rate modeling, and in turn advocated model structures for both multi-curve emerging and developed market forward IBOR processes (see Definitions \ref{EM-IBOR} and \ref{DM-IBOR}, respectively). We document these multi-curve forward IBOR processes within the xy-HJM context in the next definition.

\begin{Definition}
Within the xy-HJM framework, the multi-curve emerging market $y$-tenored forward IBOR process is given by
\begin{equation}
L^y_t(T_{i-1},T_i) = \frac{1}{\delta_i}\left(v^y_t(T_{i-1},T_i) - 1\right),
\end{equation}
with the FCF process, $v^y_t(T_{i-1},T_i)$, satisfying
\begin{equation}
\frac{\rd v^y_t(T_{i-1},T_i)}{v^y_t(T_{i-1},T_i)} =  \left(\Sigma^y_{tT_i} - \Sigma^y_{tT_{i-1}}\right) \left(\Sigma^y_{tT_i} \rd t +\lambda^y_t  \rd t + \rd W^y_t \right),
\end{equation}
for $0 \le t \le T_{i-1}$, such that the process $(h_t^y P^y_{tT_i} v^y_t(T_{i-1},T_i))_{0 \le t \le T_{i-1}}$ is a $\PR$-(local) martingale. The multi-curve developed market $y$-tenored forward IBOR process is given by
\begin{equation}
L^{xy}_t(T_{i-1},T_i) = Q^{xy}_{tT_i}L^y_t(T_{i-1},T_i) = \frac{1}{\delta_i}\left(v^{xy}_t(T_{i-1},T_i) - Q^{xy}_{tT_i}\right),
\end{equation}
with the converted FCF process, $v^{xy}_t(T_{i-1},T_i):=Q^{xy}_{tT_i} v^y_t(T_{i-1},T_i)$, satisfying
\begin{equation}
\frac{\rd v^{xy}_t(T_{i-1},T_i)}{v^{xy}_t(T_{i-1},T_i)} = \left(\Sigma^x_{tT_i} - \Sigma^w_{tT_{i-1}}\right) \left(\Sigma^x_{tT_i} \rd t +\lambda^x_t \rd t + \rd W^x_t\right)  - (\Sigma^z_{tT_{i-1}} + \lambda^z_t) \rd Z_t,
\end{equation}
for $0 \le t \le T_{i-1}$, such that the process $(h_t^x P^x_{tT_i} v^{xy}_t(T_{i-1},T_i))_{0 \le t \le T_{i-1}}$ is a $\PR$-(local) martingale.
\end{Definition}

We note that $(h_t^x P^x_{tT_i} L^{xy}_t(T_{i-1},T_i))_{0 \le t \le T_{i-1}}$ is also a $\PR$-(local) martingale, however the multi-curve developed market $y$-tenored forward IBOR process does not have an elegant differential representation as it is essentially the difference between two stochastic processes, these being the converted FCF process and the curve-conversion factor process.

\begin{Remark}
The only parameter restrictions required by the xy-HJM framework are the classical HJM drift conditions for both the $x$- and $y$-market ZCB systems. Therefore model (iii), as presented in Grbac \& Runggaldier \cite{gr}, is also a viable model for the developed market forward IBOR process, albeit an unnatural one given the incompatibility between the $x$-market pricing kernel ($h^x_t$) and the $y$-market forward IBOR process ($L^y_t(T_{i-1},T_i)$). Another viable model within the xy-HJM framework is that of Nguyen \& Seifried \cite{ns}, given by equation (\ref{alt-xy-IBOR}), however recall the observations in Remark \ref{alt-XY} regarding this model.
\end{Remark}
In the next section, rational multi-curve models are introduced. Such models, and in particular those produced in Section \ref{grmcm}, provide a rich class of flexible and tractable specifications for xy-HJM multi-curve models and associated spread dynamics.

\section{Rational multi-curve models}\label{R-Sec}

As reported in Grbac \& Runggaldier \cite{gr}, multi-curve rational interest rate models based on the pricing kernel approach have appeared in Cr\'epey et al. \cite{cmns} and in Nguyen \& Seifried \cite{ns}. The multi-curve approach proposed by Cr\'epey et al. \cite{cmns} assumes a discount bond system associated with an overnight-indexed swap (OIS) market and introduces a (forward) LIBOR process that has a built-in spread when compared to the OIS rate. The OIS-based discount bond price system, which in our setup would correspond to the $x$-curve ZCB price system, is generated by pricing kernel models driven by stochastic factors. The (forward) LIBOR process is derived by pricing a forward rate agreement (FRA) written on the LIBOR. The factor-based model of the multi-curve (forward) LIBOR process is then deduced from the no-arbitrage relation the FRA price process is required to satisfy. The LIBOR model turns out to be a rational function(al) of stochastic drivers that is given in units of the OIS pricing kernel proxy. Thus, whenever the LIBOR dynamics depend on an idiosyncratic driving factor (not affecting the OIS pricing kernel proxy), an OIS-LIBOR spread is generated that depends on a spread-idiosyncratic stochastic factor. The source of the spread can be readily read off from the expression of the LIBOR model owing to the transparency of the multi-curve approach brought forward. Given that the OIS-LIBOR spread is obtained by focusing on how the offer rate is modelled, we refer to Cr\'epey et al. \cite{cmns}, and also the Nguyen \& Seifried \cite{ns}, as a {\it rate-based modelling approach}.
\subsection{Hybrid rational-LMM multi-curve models}
A feature that is rather telling in understanding the structure of multi-curve models, and thus helps in their classification, is the nature of the discount and the forecasting curve, respectively. In the multi-curve models by Cr\'epey et al. \cite{cmns}, the term-structure of the discount (OIS-based) curve is constructed by a rational model. However, the LIBOR model is postulated in a rather ad-hoc manner and ensues directly from modelling the payoff of the forward rate agreement written on it. Similar to the hybrid HJM-LMM models in Section \ref{xy-HJM-Sec}, the forecasting curve (i.e. LIBOR-based term structure) is constructed akin to LIBOR market models. This is why we refer to Cr\'epey et al. \cite{cmns}, and to some extent also to Nguyen \& Seifried \cite{ns}, as {\it rational-LMM hybrid models}. Next, we establish the relations between these models and the framework presented in this paper.
\begin{Proposition}
Let $K(t;T_{i-1},T_i)$ be the value at time $t\in[0,T_{i-1}]$ of the fair FRA rate obtained in Cr\'epey et al. \cite{cmns}, Section 2.1. Then it holds that $K(t;T_{i-1},T_i)=K^{xy}_t(T_{i-1},T_i)$, where $K^{xy}_t(T_{i-1},T_i)$ is determined by Eq. (\ref{xy-FRA-rate}).
\end{Proposition}
\begin{proof}
By setting $P_{tT_i}=P^x_{tT_i}$, it follows that
\begin{equation}\label{xy-crepey}
L(t;T_{i-1},T_i)=P^{x}_{tT_i}L^{xy}_t(T_{i-1},T_i)
		      =P^{xy}_{tT_i}L^{y}_t(T_{i-1},T_i),
\end{equation}
where $L(t;T_{i-1},T_i)$ is the LIBOR specified in Cr\'epey et al. \cite{cmns}, Eq. (2.6).
\end{proof}
Furthermore, in Section 2.2 of Cr\'epey et al. \cite{cmns}, a particular class of rational LIBOR models is presented that becomes the workhorse, later in the paper. Next we show how such class is obtained within the $xy$-framework.
\begin{Remark}
From the relation (\ref{xy-crepey}) and by recalling that $P^y_{tT_i}=\E[h^y_{T_i}\vert\F_t]/h^y_t$, we deduce that
\begin{equation}\label{crepey-mm}
L(t;T_{i-1},T_i)=\frac{1}{\delta_i\, h^{x}_t}\left(\E\left[h^{y}_{T_{i-1}}\,\big\vert\,\F_t\right]-\E\left[h^{y}_{T_{i}}\,\big\vert\,\F_t\right]\right).
\end{equation}
Next we specify the discounting and forecasting kernels as follows: 
\begin{align}
&h^{x}_t=P^{x}_{0t}+b_1(t)\,A^{(1)}_t,\label{x-pk-rational}\\
&h^{y}_t=P^{y}_{0t}+\bar{b}_2(t)\,A^{(2)}_t+\bar{b}_3(t)\,A^{(3)}_t,
\end{align}
where, for $i=1,2,3$, the processes $(A^{(i)}_t)$ are martingales. The quantities $P^{x}_{0t}$, $P^{y}_{0t}$, $b_1(t)$ and $\bar{b}_i(t)$, $i=2,3$, are suitably chosen deterministic functions. The correspondence to the rational multi-curve LIBOR models by Cr\'epey et al. \cite{cmns}, Section 2.2, is found by setting
\begin{align}\label{xy-crepey-specs}
&L(0;T_{i-1},T_i)=\frac{1}{\delta_i}\left(P^{y}_{0T_{i-1}}-P^{y}_{0T_i}\right),\nn\\
&b_2(T_i,T_{i-1})=\frac{1}{\delta_i}\left[\bar{b}_{2}(T_{i-1})-\bar{b}_{2}(T_{i})\right],
&b_3(T_i,T_{i-1})=\frac{1}{\delta_i}\left[\bar{b}_{3}(T_{i-1})-\bar{b}_{3}(T_{i})\right].
\end{align}
The specifications (\ref{xy-crepey-specs}) cause a slight loss of generality. However, whether in practical terms such specifications are indeed restrictive can be decided once this model class is calibrated to actual market data.
\end{Remark}
We now turn our attention to the rational multi-curve models presented in Nguyen \& Seifried \cite{ns}. They propose to make use of the so-called FX-analogy to motivate pricing kernel models for the spread observed between the OIS rate and LIBOR. In particular in Section 4, Theorem 4.1, a multiplicative spread is considered. The spread is given by the ratio of a conditional expectation of the OIS-based pricing kernel (state-price deflator) and a conditional expectation of a hypothetical pricing kernel. The latter deflator may be associated with a foreign currency, although they distance themselves from such an interpretation, c.f. Section 4.2 of Nguyen \& Seifried \cite{ns}. It is our view that, although the OIS-LIBOR spread is interpreted as a kind of currency exchange rate in their work, the deduced rational multi-curve LIBOR models are of the kind that Cr\'epey et al. \cite{cmns} derive. This is especially so because the rational LIBOR models developed in Nguyen \& Seifried \cite{ns} are {\it rate-based models}---just as those produced by Cr\'epey et al. \cite{cmns}---which relate the OIS forward rate and LIBOR, directly. As they seek to dissociate themselves from the work of Bianchetti \cite{bianchetti}, who, among other authors, unequivocally sticks to the FX-analogy, we show that the rational models in Nguyen \& Seifried \cite{ns} do not need to rely on the FX-analogy in order to be derived within a pricing kernel setup.   
\begin{Proposition}
In Nguyen \& Seifried \cite{ns}, the multi-curve fair FRA rate $L^{\Delta}(t;T,T+\Delta)$ is given by 
\begin{equation}\label{L-ns}
L^{\Delta}(t;T,T+\Delta)=\frac{1}{\Delta}\left(\frac{p(t,T)}{p(t, T+\Delta)}\frac{\E\left[D_T^{\Delta}\vert\,\F_t\right]}{\E\left[D_T\vert\,\F_t\right]}-1\right),
\end{equation}
for $t\in[0,T]$. This model can be obtained by the following specification of the LIBOR process $(L(t;T_{i-1},T_i))_{0\le t\le T_{i-1}}$, for $i=1,2,\ldots,n$, in Cr\'epey et al. \cite{cmns}, Section 2.1, Eq. (2.7):
\begin{equation}\label{crepey-ns}
L(t;T_{i-1},T_i)=\frac{1}{\Delta D_t}\left(\E\left[D^{\Delta}_{T}\,\vert\,\F_t\right]-\E\left[D_{T+\Delta}\,\vert\,\F_t\right]\right),
\end{equation}
where $T_{i-1}=T$ and $T_i=T+\Delta$.
\end{Proposition}
\begin{proof}
Relation (\ref{crepey-ns}) is directly obtained by equating the fair FRA rate (4.2) in Nguyen \& Seifried \cite{ns} with the fair FRA rate (2.7) in Cr\'epey et al. \cite{cmns}. This shows that the OIS-LIBOR spread models, given in Theorem 4.1 in Nguyen \& Seifried \cite{ns}, do not necessitate the use of the FX-analogy in order to derive (rate-based) multi-curve discounting models in a pricing kernel approach. While $(D_t)$ corresponds to the OIS-associated pricing kernel process $(\pi_t)$ in Cr\'epey et al. \cite{cmns}, there is indeed no reason to identify the process $(D^{\Delta}_t)$ with a fictitious pricing kernel associated with a foreign currency/economy. It may just be viewed as an idiosyncratic component of the LIBOR process. 
\end{proof}
\begin{Remark}
Comparing Eq. (\ref{crepey-ns}) with Eq. (\ref{crepey-mm}) we observe a discrepancy in the way that the conversion to a multi-curve setup is obtained in Nguyen \& Seifried \cite{ns}. The source of such incongruence is discussed in Remark \ref{alt-XY}, (i). The difference is resolved by the following adjustment in the multi-curve model (\ref{L-ns}):
\begin{equation}
L^{\Delta}(t;T,T+\Delta)=\frac{1}{\Delta}\left(\frac{p(t,T)}{p(t, T+\Delta)}\frac{\E\left[D_T^{\Delta}\vert\,\F_t\right]}{\E\left[D_T\vert\,\F_t\right]}-1\cdot Q(t,T+\Delta)\right),
\end{equation}
where, based on to the $xy$-approach, the conversion factor $Q(t,T+\Delta)$, or spread process, is given by
\begin{equation}
Q(t,T+\Delta)=\frac{\E\left[D_{T+\Delta}^{\Delta}\vert\,\F_t\right]}{\E\left[D_{T+\Delta}\vert\,\F_t\right]}.
\end{equation} 
The adjustment allows the model to be derived by a consistent application of the FX-analogy in a pricing kernel setup as produced in the $xy$-approach developed in this paper.
\end{Remark}

\subsection{Pure-rational multi-curve models}\label{grmcm}

Unlike the preceding rational-LMM hybrid multi-curve models, we now consider rational models for both, the discounting curve and the forecasting curve, that is for the ZCB price process $(P^{x}_{tT_i})_{0\le t\le T_i}$ and $(P^{y}_{tT_i})_{0\le t\le T_i}$, respectively, which feature as desirable properties (i) tractability, (ii) transparency of the dependence structure among the risk factors and thus (iii) a good understanding of the resulting model for the spread dynamics between the $x$- (discounting) and the $y$- (forecasting) curves. The rational price models considered by Macrina \cite{am}, and by Cr\'epey et al. \cite{cmns} for multi-curve interest rate modelling in particular, offer the set of properties we require. For the $x$- and $y$-ZCB, we postulate the following:
\begin{align}\label{x&y-zcb}
&P^{x}_{tT_i}
=\frac{P^{x}_{0T_i}\prod^m_{k=1} Z^{x}_k(t,T_i)}{P^{x}_{0t}\prod^m_{k=1} Z^{x}_k(t)},&
&P^{y}_{tT_i}
=\frac{P^{y}_{0T_i}\prod^n_{\ell=1} Z^{y}_\ell(t,T_i)}{P^{y}_{0t}\prod^n_{\ell=1} Z^{y}_\ell(t)},&
\end{align}
where $Z^{x}_k(t,T_i)=(1+b^{x}_k(T_i)\,A^{x}_{t,\,k})$ and $Z^{y}_{\ell}(t,T_i)=(1+b^{y}_{\ell}(T_i)\,A^{y}_{t,\,\ell})$ are taken to be positive processes. The quantities $P^{x}_{0t}$ and $P^{y}_{0T_i}$ are the initial term structures of the $x$ and $y$ ZCBs, $b_{k}$ and $b_{\ell}$ are deterministic functions, and $(A^{x}_{t,\,k})$ and $(A^{y}_{t,\,\ell})$ are martingales with respect to some $(\PR$-equivalent) probability measure. For further (technical) details, we refer to Macrina \cite{am} and Cr\'epey et al. \cite{cmns}. We take a closer look at $(P^{y}_{tT_i})$, although the structural properties of the model also apply to $(P^{x}_{tT_i})$. The return process of the forecasting ZCB is given by 
\begin{equation}
\ln\left(P^{y}_{tT_i}\right)=\ln\left(\frac{P^{y}_{0T_i}}{P^{y}_{0t}}\right)+\sum^n_{\ell=1}\ln\left(\frac{1+b^{y}_{\ell}(T_i)\,A^{y}_{t,\,\ell}}{1+b^{y}_{\ell}(t)\,A^{y}_{t,\,\ell}}\right).
\end{equation}
The associated short rate process $(r^{y}_t)$ is given by
\begin{equation}\label{n-factor_y-rate}
r^{y}_t=-\left(\frac{\partial_t P^{y}_{0t}}{P^{y}_{0t}}+\sum^n_{\ell=1}\theta^{y}_{t,\,\ell}\right),
\end{equation}
where we define the $(A^{y}_{t,\,\ell})$-driven factor component $(\theta^{y}_{t,\,\ell})$ by 
\begin{equation}\label{theta-ry}
\theta^{y}_{t,\,\ell}=\frac{\partial_t b^{y}_{\ell}(t)\,A^{y}_{t,\,\ell}}{1+b^{y}_{\ell}(t)\,A^{y}_{t,\,\ell}}.
\end{equation}
Now let us assume, for the sake of the explanation, that the number $n$ of factor components is given by the particular tenor $y$. So, let $a=1,2,3,\ldots$, $y=3a$ which is the 3-month, 6-month, 9-month, 12-month, etc LIBOR tenor, and $n=a+2$ the number of factor components. For the 1-month LIBOR tenor, we assume that the short rate of the associated 1-month ZCB is driven by two factor components, i.e $n=2$. Then, we have the following additive structure for the short rate model associated with the corresponding forecasting ZCBs:
\begin{align}\label{short-rate-structure}
&y=1\text{-month-tenored ZCB},\ (P^{1}_{tT_i}):& &r^{1}_t=-\left(\frac{\partial_t P^{1}_{0t}}{P^{1}_{0t}}+\theta^{1}_1(t)+\theta^{1}_2(t)\right),&\nn\\
&y=3\text{-month-tenored ZCB},\ (P^{3}_{tT_i})
:& &r^{3}_t=-\left(\frac{\partial_t P^{3}_{0t}}{P^{3}_{0t}}+\theta^{3}_1(t)+\theta^{3}_2(t)+\theta^{3}_3(t)\right),&\nn\\
&\vdots & &\vdots &\nn\\ 
&y=3a\text{-month-tenored ZCB},\ (P^{y}_{tT_i}):& &r^{y}_t=-\left(\frac{\partial_t P^{y}_{0t}}{P^{y}_{0t}}+\sum^{a+2}_{\ell=1}\theta^{y}_\ell(t)\right), \quad a=1,2,3,\ldots&. 
\end{align}
Depending on the specific interbank offer rate market, we could envisage the situation where $\theta^{i}_\ell=\theta^{j}_\ell$ for all $i,j=3a$. This would mean that the various $y$-curves only differed by the number of factor components driving the corresponding short rates (i.e, forecasting ZCBs). We would then have
\begin{align}
&y\text{-month-tenored ZCB},\ (P^{y}_{tT_i}):& &r^{y}_t=-\left(\frac{\partial_t P^{y}_{0t}}{P^{y}_{0t}}+\sum^{a+2}_{\ell=1}\theta^{}_\ell(t)\right), \quad a=1,2,3,\ldots,&
\end{align}
where the short rate model of the 1m-tenored ZCB is recovered by setting $a=0$. From the FRA price process (\ref{xyquanto-FRA}), one sees that the quantity responsible for the consistent transfer from a single-curve to a multi-curve setting is the quanto-bond with price process $(P^{xy}_{tT_i})$. 
\

Next we introduce the multiplicative class of rational models for the ``multi-curve'' quanto-bond price process $(P^{xy}_{tT_i})_{0\le t\le T_i}$:
\begin{equation}
P^{xy}_{tT_i}=\frac{1}{h^{x}_t}\E\left[h^{y}_{T_i}\,\big\vert\,\F_t\right]=\frac{P^{y}_{0T_i}\prod^n_{\ell=1} Z^{y}_\ell(t,T_i)}{P^{x}_{0t}\prod^m_{k=1} Z^{x}_k(t)}.
\end{equation} 
The model for the short rate of interest $(r_t^{xy})_{0\le t}$, associated with the  quanto-bond, is obtained by $r^{xy}_t=-\partial_{T_i}\ln(P^{xy}_{tT_i})\vert_{T=t}$, assuming that the quanto-bond price function is differentiable in its maturity $T_i$. It follows that $r_t^{xy}=r^{y}_t$. One could argue that it is somewhat artificial to introduce the $x$-discounting bond because, after all, the $x$-curve may be a specific $y$-curve. We wish however to allow for more generality: there is no reason why the type of model ought to be the same for the $x$-ZCB and for the $y$-ZCB. It is only for convenience that we here decide to consider the same type of pricing model for both types of bonds. In any case, the discounting curve---identified with the one-day deposit---can be viewed as the $y=0$-forecasting curve in the above setup (\ref{short-rate-structure}):
\begin{align}
&\text{1-day-tenored ZCB},\ P^{x}_{tT_i}=P^{0}_{tT_i}:& &r^{x}_t=r^{0}_t=-\left(\frac{\partial_t P^{0}_{0t}}{P^{0}_{0t}}+\theta^{0}_1(t)\right).&
\end{align}  
A byproduct of the multi-curve modelling approach based on bonds as considered in this paper, is the implicit, or rather emerging, spread models. Within the rational models, the process for the spread between the $y$ and $y+3\textrm{m}$ curves is given by
\begin{equation}
s^{y,y+3m}_{tT_i}=\frac{P^{y+3m}_{tT_i}}{P^{y}_{tT_i}}=\frac{P^{y}_{0t}}{P^{y+3m}_{0t}}\frac{P^{y+3m}_{0T_i}}{P^{y}_{0T_i}}\,\Delta_{a+2}(t,T_i), \quad a=1,2,3,\ldots\ ,
\end{equation}
where the stochastic spread process $(\Delta_{a+2}(t,T_i))_{0\le t\le T_i}$ is given by 
\begin{equation}
\Delta_{a+2}(t,T_i)=\frac{1+b^{y}_{a+2}(T_i)\,A^{y}_{t,\,a+2}}{1+b^{y}_{a+2}(t)\,A^{y}_{t,\,a+2}}.
\end{equation}
We note that the stochastic spread is positive assuming that the rates underlying the tenors are non-negative, see Corollary \ref{Corr-FCF}.
\subsection{Linear-rational term structure models}
Filipovi\'c et al. \cite{flt} introduce the so-called linear-rational term structure (LRTS) models. In this section we show how the multi-curve extension to the LRTS is produced by showing that the LRTS models belong to the more general class introduced in the previous section. We thus prove that (a) the LRTS models belong to the class of models developed in Macrina \cite{am} when an infinite-time horizon is considered, and (b) that the pricing kernel generating the LRTS is a weighted heat kernel (WHK). Pricing kernels generated by WHKs in an infinite time horizon setting are introduced in Akahori et al. \cite{ahtt} and developed in Akahori \& Macrina \cite{akm} in the case tha the WHK is driven by a time-inhomogeneous Markov process. In particular, we shall show that the LRTS models produce bond price processes $(P_{tT})_{0\le t\le T}$ of the form
\begin{equation}\label{bA-models}
P_{tT}=\frac{P_{0T}+b(T)\,A_t}{P_{0t}+b(t)\,A_t},\qquad (0\le t\le T)
\end{equation}
which are identified as a class of Markov functionals. The function $b(t)$, $0\le t\le T$, is deterministic and $(A_t)_{0\le t}$ is a martingale process. The explicit construction of this class of term structure models is presented in Macrina \cite{am}.
\begin{Definition}\label{Def-LRTS} Linear-Rational Term Structure Models, Filipovi\'c et al. \cite{flt}.
\\
Let $(Z_t)_{0\le t}$ denote the multivariate process with state space $E\subset\R^m$ that satisfies the stochastic differential equation
\begin{equation}\label{Z-SDE}
\rd Z_t =\kappa(\theta-Z_t)\rd t + \rd M_t,
\end{equation}
where $\kappa\in\R^{m\times m}$ and $\theta\in\R^m$, and where $(M_t)_{0\le t}$ is an $m$-dimensional martingale. Let $(\zeta_t)_{0\le t}$ denote the pricing kernel process defined by
\begin{equation}\label{LRTS-pk}
\zeta_t=\e^{-\alpha t}\left(\phi+\psi Z_t\right),
\end{equation}
where $\alpha\in\R$, $\phi\in\R$, and $\psi\in\R^m$ such that $\phi+\psi z>0$ for all $z\in E$. The linear-rational term structure, generated by the linear pricing kernel process $(\zeta_t)_{0\le t}$, have zero-coupon bond price processes $(P_{tT})_{0\le t\le T}$ given by
\begin{equation}\label{LRTS-P}
P_{tT}=\e^{-\alpha (T-t)}\frac{\phi + \psi\theta + \psi\e^{(T-t)}\left(Z_t-\theta\right)}{\phi + \psi Z_t},
\end{equation}
where $T$ is the bond maturity date.
\end{Definition}
\begin{Proposition}
The stochastic differential equation (\ref{Z-SDE}) has the unique solution given by
\begin{equation}\label{Z-sol}
Z_t=\e^{-\kappa t}\left(Z_0+\kappa\int^t_0 \e^{\kappa s}\rd s\,\theta\right)+\e^{-\kappa t}A_t.
\end{equation} 
The process $(A_t)_{0\le t}$, defined by
\begin{equation}\label{A-mart}
A_t=\int^t_0\e^{\kappa s}\rd M_s,
\end{equation}
is a martingale.
\end{Proposition}
\begin{proof}
That the mean-reverting process (\ref{Z-sol}) is the unique solution to the SDE (\ref{Z-SDE}) follows from a straightforward application of Ito's Lemma. To show that $(A_t)_{0\le t}$ is a martingale, one remarks that $\E[|\,A_t|]<\infty$ for all $t\ge 0$ and that $\E[A_u\vert\F_s]=A_s$, for $0\le s\le u$. The latter follows by calculating $\E[\int^u_s \rd[\phi(t)M_t]\,\vert\,\F_s]$, where $0\le s\le t\le u$, and by applying Fubini's theorem. One then obtains 
\begin{equation}
\E[A_u\vert\F_s]-A_s=\E[\phi(u)M_u-\phi(s)M_s\,\vert\,\F_s]-\E\left[\int^u_s M_t\partial_t\phi(t)\rd t\,\big\vert\,\F_s\right]=0,
\end{equation}
which completes the proof.
\end{proof}
\begin{Theorem}
The pricing kernel process $(\zeta_t)_{0\le t}$ that generates the linear-rational term structure models, specified in Definition \ref{Def-LRTS}, is given by 
\begin{equation}
\zeta_t=\zeta_0\left[P_{0t}+b(t)\,A_t\right],
\end{equation}
where $\zeta_0=\phi+\psi Z_0$. The positive, deterministic function $(P_{0t})_{0\le t\le T}$ is the initial term structure of the associated $T$-maturity bond system with price process
\begin{equation}\label{bA-P}
P_{tT}=\frac{P_{0T}+b(T)\,A_t}{P_{0t}+b(t)\,A_t}\qquad (0\le t\le T)
\end{equation}
where $P_{0t}$, the deterministic function $b(t)$ and the martingale $(A_t)$ are determined by
\begin{eqnarray}\label{LRTS-P0t}
P_{0t}&=&\frac{\e^{-\alpha t}}{\phi+\psi Z_0}\left[\phi+\psi\e^{-\kappa t}\left(Z_0+\kappa\int^t_0\e^{\kappa s}\rd s\, \theta\right)\right], \qquad 0\le t\le T,\label{P0t}\\
b(t)&=&\frac{\e^{-\alpha t}}{\phi+\psi Z_0}\,\psi\e^{-\kappa t}, \qquad 0\le t\le T,\label{bt}\label{LRTS-b}\\
A_t&=&\int^t_0\e^{\kappa s}\rd M_s, \qquad t\ge 0.\label{At}\label{LRTS-A}
\end{eqnarray}
\end{Theorem}
\begin{proof}
One direction is straightforward: it suffices to insert (\ref{P0t}), (\ref{bt}) and (\ref{At}) in (\ref{bA-P}) to obtain (\ref{LRTS-P}). The other direction, i.e. beginning from Definition \ref{Def-LRTS}, goes as follows: The solution (\ref{Z-sol}) is inserted in (\ref{LRTS-P}) to obtain
\begin{equation}
P_{tT}=\frac{\e^{-\alpha T}\left[\phi+\psi\e^{-\kappa T}\left(Z_0+\kappa\int^t_0\e^{\kappa s}\rd s\,\theta\right)\right]+\e^{-\alpha T}\psi\kappa\int^t_0\e^{-\kappa(T-s)}\rd s\,\theta+\e^{-\alpha T}\psi\e^{-\kappa T}A_t}{\e^{-\alpha t}\left[\phi+\psi\e^{-\kappa t}\left(Z_0+\kappa\int^t_0\e^{\kappa s}\rd s\,\theta\right)\right]+\e^{-\alpha t}\psi\e^{-\kappa t}A_t}.
\end{equation}
Next, we define the functions $\gamma(t,T)$, $\lambda(t,T)$ and $\tilde{b}(t)$ by
\begin{eqnarray}
\gamma(t,T)&=&\e^{-\alpha T}\left[\phi+\psi\e^{-\kappa T}\left(Z_0+\kappa\int^t_0\e^{\kappa s}\rd s\,\theta\right)\right],\\
\lambda(t,T)&=&\e^{-\alpha T}\psi\kappa\int^t_0\e^{-\kappa(T-s)}\rd s\,\theta,\\
{\tilde{b}(t)}&=&\e^{-\alpha t}\psi\e^{-\kappa t}
\end{eqnarray}
for $t\in[0,T]$, and therewith express the bond price process in the form
\begin{equation}
P_{tT}=\frac{\gamma(t,T)+\lambda(t,T)+\tilde{b}(T)\,A_t}{\gamma(t,t)+\tilde{b}(t)\,A_t}.
\end{equation}
The initial term structure $P_{0t}$, $0\le t\le T$, satisfies the relation
$\gamma(0,0)P_{0t}=\gamma(0,t)+\lambda(0,t)=\gamma(t,t).$
Furthermore, $\gamma(t,T)+\lambda(t,T)-[\gamma(0,T)+\lambda(0,T)]=0$ holds. We thus write 
\begin{equation}
P_{tT}=\frac{\gamma(t,T)+\lambda(t,T)-[\gamma(0,T)+\lambda(0,T)]+\gamma(0,0)P_{0T}+\tilde{b}(T)\,A_t}{\gamma(0,0)P_{0t}+\tilde{b}(t)\,A_t},
\end{equation}
and immediately obtain (\ref{bA-P}) by observing that $b(t)=\tilde{b}(t)/\gamma(0,0)$ for $0\le t\le T$.
\end{proof}
\begin{Corollary}
The Linear-Rational Term Structure models can be expressed in the form 
\begin{equation}\label{barbA-P}
P_{tT}=\frac{P_{0T}\left[1+\bar{b}(T)\,A_t\right]}{P_{0t}\left[1+\bar{b}(t)\,A_t\right]},
\end{equation}
for $0\le t\le T$, where $\bar{b}(t)=b(t)/P_{0t}$. This is the form (\ref{x&y-zcb}) for $m=1$, and thus the necessary basis for the extension to the {\it multi-curve linear-rational term structure models via Theorem \ref{xy-IBOR-FRA-theorem} and Definition \ref{xy-IBOR}, in a developed market, and via Definitions \ref{EM-IBOR} and \ref{quanto-bond} in the emerging market}.
\end{Corollary}
\begin{proof}
This follows directly from (\ref{bA-P}). 
\end{proof}
\begin{Remark}
We emphasise that the form (\ref{bA-P}), or equivalently (\ref{barbA-P}), shows that the Linear-Rational Term Structure has, by (\ref{LRTS-P0t}), a functionally fully specified initial term structure $P_{0t}$ of bond prices for $t\in[0, T]$. Also, the models (\ref{bA-P}) specified by (\ref{LRTS-P0t})-(\ref{LRTS-A}) produce an example of the larger class (\ref{bA-models}), or equivalently (\ref{barbA-P}), of term structure models that can accommodate unspanned stochastic volatility as considered in Filipovi\'c et al. \cite{flt}, Section C.
\end{Remark}
Next we consider weighted heat kernel processes over an infinite-time horizon, see Akahori et al. \cite{ahtt}, and in particular the case where the propagator is a conditional expectation, as in Akahori \& Macrina \cite{akm} and Macrina \cite{am}. Such weighted heat kernels are used to generate (explicit) pricing kernel processes. The definition that follows provides weighted heat kernels in a multivariate setting. 
\begin{Definition} Let $(X_t)_{0\le t}$ be an $m$-dimensional $(\F_t)$-adapted Markov process, $F(t,x)$ be a vector-valued and deterministic function in $\R^m$, and $w(t,u)$ a matrix-valued deterministic function in $\R^{m\times m}$. Furthermore, let the functions $f_0(t)\in\R$ and $f_1(t)\in\R^m$ be deterministic. The process $(\pi_t)_{0\le t}$ is a weighted heat kernel defined by
\begin{equation}\label{whk}
\pi_t=f_0(t)+f_1(t)\int^{\infty}_0 w(t,u)\E\left[F(t+u,X_{t+u})\,\vert\,\F_t\right]\rd u,
\end{equation}
where $t\wedge u\ge 0$, and $f_0(t)$, $f_1(t)$, $F(t,x)$ and $w(t,u)$ are chosen such that $(\pi_t)$ is a positive and finite (scalar-valued) process.
\end{Definition}
The next statement asserts that the pricing kernel process $(\zeta_t)$ in Filipovi\'c et al. \cite{flt} is a weighted heat kernel and it establishes the relation between $(\zeta_t)$ and the class (\ref{whk}).
\begin{Theorem}
The pricing kernel (\ref{LRTS-pk}) that generates the linear-rational term structure models by Filipovi\'c et al. \cite{flt}, is a special case of the process (\ref{whk}) where the following holds:
\begin{enumerate}
\item Let $(X_t)$ be the Markov process $(Z_t)$ that satisfies (\ref{Z-SDE}).
\item $F(t,X_t)=Z_t$, for all $t\ge 0$.
\item $w(t,u)=\e^{-\beta(t+u)}$, $\beta\in\R^{m\times m}$ invertible where $\beta\kappa=\kappa\beta$ for $\kappa\in\R^{m\times m}$ invertible.
\item The functions $f_0(t)$ and $f_1(t)$ are give by
\begin{eqnarray}\label{f0}
f_0(t)&=&f_1(t)\left[(\beta+\kappa)^{-1}-\beta^{-1}\right]\e^{-\beta t}\theta+\e^{-\alpha t}\phi,\\
f_1(t)&=&\e^{-\alpha t}\psi\e^{\beta t}(\beta+\kappa),\label{f1}
\end{eqnarray}
where $\phi\in\R$, $\alpha\in\R$, $\theta\in\R^m$, $\psi\in\R^m$ and $\beta\in\R^{m\times m}$ with $\beta\kappa=\kappa\beta$ for $\kappa\in\R^{m\times m}$. It is assumed that $(\beta+\kappa)$ is invertible.
\end{enumerate}
\end{Theorem}
\begin{proof}
One direction is straightforward: It suffices to insert items 1-\,4 into Equation (\ref{whk}) to obtain the pricing kernel process (\ref{LRTS-pk}). In the other direction, that is starting from (\ref{whk}), one makes the initial assumptions that the first and second items shall hold. This leads to 
\begin{equation}
\E\left[Z_{t+u})\,\vert\,\F_t\right]=\e^{-\kappa(t+u)}\left(Z_0+\kappa\int^{t+u}_0\e^{\kappa s}\rd s\,\theta\right)+\e^{-\kappa(t+u)}\int^t_0\e^{\kappa s}\rd M_s.
\end{equation}
Then, by choosing the ansatz given in the third item, one obtains
\begin{equation}
\int^{\infty}_0 w(t,u)\E\left[F(t+u,X_{t+u})\,\vert\,\F_t\right]\rd u=(\beta+\kappa)^{-1}\e^{-\beta t}Z_t+\left[\beta^{-1}-(\beta+\kappa)^{-1}\right]\e^{-\beta t}\theta.
\end{equation}
Thus, the functions $f_0(t)$ and $f_1(t)$ are selected such that the pricing kernel process (\ref{LRTS-pk}) is obtained, that is (\ref{f0}) and (\ref{f1}).
\end{proof}

\section{Inflation-linked and FX pricing}\label{IL-FX}

In this section, we show how the across-curve valuation approach developed in this paper extends to the pricing of other fixed-income financial instruments. The curve-conversion factor process, developed in the present work, may conveniently be applied to the pricing and hedging of inflation-linked and foreign-exchange (FX) securities. In particular, the quanto-bond process $(P^{xy}_{tT})_{0\le t\le T}$ plays an important role in the pricing of hybrid securities, suchlike inflation-linked foreign-exchange products, where consistent asset valuation can still be a challenge. 
\subsection{Inflation-linked pricing}
It is customary in inflation-linked price modelling and hedging to consider two economies, the so-called nominal and real economies. Such a viewpoint matches the $xy$-concept so much so that the curve-conversion factor associated with inflation-linked pricing is obtained with little effort. But this is the strength and appeal we see in this formalism. The nominal  (N), cash-based economy is associated with the $x$-curve, and the real (R), goods/services-based economy is associated with the $y$-curve. So, we set $x=N$ and $y=R$. Next we apply the scheme developed in Sections \ref{CCF} and \ref{PK-Multi} of this paper. 
\

We assume positive pricing kernel processes $(h^N_t)_{0\le t}$ and $(h^R_t)_{0\le t}$ for the nominal and the real economies, respectively. The process $(C_t)_{0\le t}$ of the consumer price index links prices between the nominal and the real economies by
\begin{equation}\label{PLP}
C_t=C_0\frac{h^R_t}{h^N_t},
\end{equation}
where $C_0$ is the base price level at time 0 (not necessarily normalised to one). The price $P^{NR}_{tT}$ at time $t\le T$ of an inflation-linked discount bond with cash flow $C_T$ at maturity $T$ is given by
\begin{equation}\label{PNR-price}
P^{NR}_{tT}=\frac{1}{h^N_t}\,\E\left[h^N_T \frac{C_T}{C_0}\,\vert\,\F_t\right]=\frac{1}{h^N_t}\,\E\left[h^R_T\,\vert\,\F_t\right].
\end{equation}
In the $xy$-formalism, where we recall $x=N$ and $y=R$, we may write the price process $(P^{NR}_{tT})$ in terms of the conversion formula
\begin{equation}\label{PNR-Q}
P^{NR}_{tT}=P^{N}_{tT}\,Q^{NR}_{tT},
\end{equation} 
where $(P^{N}_{tT})$ is the price process of the nominal discount bond, and where
\begin{equation}\label{NR-convfact}
Q^{NR}_{tT}=\frac{\E\left[h^R_T\,\vert\,\F_t\right]}{\E\left[h^N_T\,\vert\,\F_t\right]}\qquad (0\le t\le T)
\end{equation}
is the curve-conversion factor (spread process) linking discounting on the nominal $N$-curve and forecasting on the real $R$-curve. The expression for the quanto-bond (\ref{PNR-price}) can be obtained in a straightforward fashion from Eq. (\ref{xy-generic}) by setting
t = T, thereafter replacing the pricing time $s$ with $t$, and further by setting $x=N$, $y=R$ and $H^R_T=1$. The nominal curve serves as the base-curve; hence the curve-conversion factor process (\ref{NR-convfact}), in the relation (\ref{PNR-Q}), quantifies the number of positions in the nominal $T$-maturity discount bond necessary to replicate the no-arbitrage value at $t\in[0,T]$ of the inflation-linked discount bond with value $P^{NR}_{tT}$ at time $t$. Given that the nominal discount bond $P^{N}_{tT}$ and the inflation-linked discount bond $P^{NR}_{tT}$ are traded sufficiently on a market, one can imply from the market the inflation-linked conversion factor 
\begin{equation}
Q^{NR}_{tT}=\frac{P^{NR}_{tT}}{P^{N}_{tT}}.
\end{equation} 
The pricing formulae for an inflation-linked forward rate agreement (or inflation-linked zero-coupon swap) and for a year-on-year swap contract can be expressed in terms of the conversion factor. The derivations of such pricing formulae follow those for the forward rate agreement and the swap contracts presented in Section \ref{PK-Multi}. Price models for inflation-linked securities, which are based on explicit pricing kernel models---hence, on explicit curve conversion factor processes---have been developed by Dam et al. \cite{DMSS}. Such models feature a high degree of flexibility and good calibration properties.
\subsection{Exchange to foreign currencies}\label{efc}
We consider two currencies $i$ and $j$ in the respective nominal (cash-based) economies $N_i$ and $N_j$. Here we show that the {\it forward foreign exchange rate}, which converts an amount of domestic currency $j$ into an amount of foreign currency $i$ at a fixed future data, is given by today's spot exchange rate multiplied with the appropriate currency conversion factor. We set $x=i$ and $y=j$ in the xy-formalism, see Sections \ref{CCF} and \ref{PK-Multi}. In the following, we abbreviate ``foreign exchange" with ``FX".
\

We denote by $(X^{ij}_t)_{0\le t}$ the spot FX rate process, which converts, e.g., GBP to EUR. We emphasise that the notation $ij$ implies, in this example, ${\rm EUR}/{\rm GBP}$. By $(F^{ij}_{tT})_{0\le t\le T}$ we denote the process of the forward FX rate. We conjecture the following relation:
\begin{equation}\label{FX-Fwd-conj}
F^{ij}_{tT}=X^{ij}_t\, \frac{P^j_{tT}}{P^i_{tT}}.
\end{equation}
Here, $(P^i_{tT})_{0\le t\le T}$ and $(P^j_{tT})_{0\le t\le T}$ are assumed to be the nominal OIS discount bond price processes denominated in the $i$ (EUR) and $j$ (GBP) currencies, respectively. We acknowledge that the correct discount bond price processes, in practice, are those determined by the respective FX basis curves. While these may be easily incorporated into the framework via pricing kernels and associated curve-conversion factor processes, we ignore this fact throughout this section for ease of exposition. We note that $F^{ij}_{tt}=X^{ij}_t$, $t\in[0,T]$. The $i$- and $j$-denominated economies are assumed to be equipped with the respective (nominal) pricing kernel processes $(h^i_t)$ and $(h^j_t)$.
By recalling the price formula of a discount bond, it follows from the conjecture (\ref{FX-Fwd-conj}) that
\begin{equation}\label{FX-fwd}
F^{ij}_{tT}=X^{ij}_t\,\frac{h^i_t\ \E[h^j_T\,\vert\,\F_t]}{h^j_t\ \E\left[h^i_T\,\vert\,\F_t\right]}
=X^{ij}_t\,\frac{h^i_t}{h^j_t}\,Q^{ij}_{tT},
\end{equation}
where the {\it FX conversion factor} $(Q^{ij}_{tT})$ for the currency pair $(i,j)$ has the familiar form
\begin{equation}\label{fwd-fx}
Q^{ij}_{tT}=\frac{\E[h^j_T\,\vert\,\F_t]}{\E\left[h^i_T\,\vert\,\F_t\right]}.
\end{equation}

Next we validate the conjecture (\ref{FX-Fwd-conj}) by pricing an FX forward contract in this setup. We model the spot FX rate process $(X^{ij}_t)$ by
\begin{equation}\label{spot-fx}
X^{ij}_t=X^{ij}_0\ \frac{h^j_t}{h^i_t},
\end{equation}
and, by recalling (\ref{FX-fwd}), we obtain
\begin{equation}\label{fwd-FX-Q}
F^{ij}_{tT}=X^{ij}_t\,\frac{h^i_t}{h^j_t}\,Q^{ij}_{tT}=X^{ij}_0\, Q^{ij}_{tT}.
\end{equation}
This is the relation we would expect to emerge in the xy-approach for the forward FX process. The stochastic price dynamics of the forward FX contract are determined by the ratio of the forecasting curves in the two economies denominated in units of the respective currencies. We shall now see whether the expression (\ref{fwd-FX-Q}) for the forward FX rate is indeed the {\it fair rate} obtained from pricing the FX forward contract. 
\begin{Proposition}
Let $(P^i_{tT})_{0\le t\le T}$ and $(P^j_{tT})_{0\le t\le T}$ be the price processes of the discount bonds denominated in the $i$ and $j$ currencies, respectively. Let $(X^{ij}_t)_{t\ge 0}$ be the spot FX rate process exchanging $j$ currency for $i$ currency at time $t\ge0$. Then, for $0\le t\le T$, the fair forward FX rate is given by
\begin{equation}\label{fair-fwd-FX}
F^{ij}_{tT}=X^{ij}_0\, Q^{ij}_{tT}=X^{ij}_t\, \frac{P^j_{tT}}{P^i_{tT}},
\end{equation}
where $(Q^{ij}_{tT})_{0\le t\le T}$ is the curve-conversion process (\ref{fwd-fx}).
\end{Proposition}
Equation (\ref{fair-fwd-FX}) confirms the expression given in conjecture (\ref{FX-Fwd-conj}). Furthermore, the FX curve-conversion factor process $(Q^{ij}_{tT})_{0\le t\le T}$ can be implied from the quoted forward FX rates and the spot rates on the market, that is, 
\begin{equation}\label{fx-Q}
Q^{ij}_{tT}=\frac{F^{ij}_{tT}}{X^{ij}_0}.
\end{equation}
\begin{proof}
Consider the payoff $V^i_T=X^{ij}_T/X^{ij}_0-K^i$ of an FX forward contract, with expiry date $T>0$ and strike value $K^i$, denominated in $i$-currency. The price process $(V^{i}_{tT})$ of the FX forward contract is given by
\begin{equation}\label{FX-fwd-contract}
V^{i}_{tT}=\frac{1}{h^i_t}\,\E\left[h^i_T\left(X^{ij}_T/X^{ij}_0-K^i\right)\,\big\vert\,\F_t\right].
\end{equation}
This follows as an application of the across-curve formula (\ref{xy-generic}), where one sets $x=i$ and $y=j$, alongside $t=T$ and where the pricing time $s$ is replaced with $t$. Furthermore, $H^j_T=1-K^i/Q^{ij}_{TT}$ shall hold, which is indeed a $j$-currency quantity. The relation $H^{ij}_{tT}=V^i_{tT}$ is obtained where we drop the $j$ superscript in $V^i_{tT}$ to emphasise that the value $V^i_{tT}$ at time $t\in[0,T]$ is given in units of the $i$-currency. Then, by recalling Eq. (\ref{spot-fx}), it follows with ease that 
\begin{equation}
V^{i}_{tT}=P^{ij}_{tT}-K^i P^i_{tT},
\end{equation}
for $t\in[0,T]$.
By setting $V^{i}_{tT}=0$, for all $t\in[0,T]$, we obtain the result stated in the proposition, where $K^i_{tT}=F^{ij}_{tT}/X^{ij}_0=Q^{ij}_{tT}$ is the fair (strike) value for the forward currency-exchange process.
\end{proof}
\subsection{Multi-curve interest rate foreign-exchange hybrid}\label{MCIRFX}
We consider the situation whereby an investor wishes to enter a foreign-currency forward contract on LIBOR. An example might clarify this type of hybrid security. Suppose we consider GBP-based LIBOR, as quoted in the U.K. market. An investor wishes to enter a USD-denominated forward contract written on the GBP-based LIBOR quotes. This exposes the investor to the risk underlying the GBP-LIBOR market and the currency-exchange risk between GBP \& USD. We are interested in deducing the fair forward rate process of the USD-forward contract based on GBP-LIBOR. This is obtained with ease in the xy-approach by combining the results in Sec. \ref{dsdm}, on multi-curve systems in developed markets, and in Sec. \ref{efc}, above.
\begin{Proposition}
Consider $0\le t\le T_{i-1}<T_i$ where $T_i-T_{i-1}$ is the tenor of the GBP-based LIBOR and $T_i$ is the expiry date of the USD-denominated forward contract that is written on the GBP-based LIBOR. The fair forward rate process $K^{x_{\$}\,y_{\pounds}}_t(T_{i-1},T_i)_{0\le t\le T_{i-1}}$ of the USD-denominated forward contract is given by
\begin{equation}\label{K-FX-LIBOR}
K^{x_{\$}\,y_{\pounds}}_t(T_{i-1},T_i)=\frac{F^{{\$}{\pounds}}_{tT_i}}{X^{\$\pounds}_0}L^{x_{\pounds}\,y_{\pounds}}_t(T_{i-1},T_i)=L^{x_{\$}\,y_{\pounds}}_t(T_{i-1},T_i),
\end{equation}
where $X^{\$\pounds}_0$ is the spot USD/GBP exchange rate at $t=0$, the fair forward USD/GBP exchange rate process $(F^{{\$}{\pounds}}_{tT_i})_{0\le t\le T_i}$ is given by the relation (\ref{fair-fwd-FX}) and the GBP-based LIBOR process $L^{x_{\pounds}\,y_{\pounds}}_t(T_{i-1},T_i)_{0\le t\le T_{i-1}}$ is given by Eq. (\ref{xy-IBOR-mart}).
\end{Proposition}
\begin{proof} The starting point is the $y$-tenored GBP-based LIBOR, given in Lemma \ref{xy-IBOR-mart-lemma}, which we convert, at time $T_{i-1}$, into the $x_{\$}$-market by the FX conversion factor (\ref{fwd-fx}). We obtain 
\begin{equation}\label{FX-LIBOR}
L^{x_{\$}\,y_{\pounds}}_{T_{i-1}}(T_{i-1},T_i)=Q^{\$\pounds}_{T_{i-1}T_i}L^{x_{\pounds}\,y_{\pounds}}_{T_{i-1}}(T_{i-1},T_i)=Q^{x_{\$}\,y_{\pounds}}_{T_{i-1}T_i}L^{y_{\pounds}}_{T_{i-1}}(T_{i-1},T_i),
\end{equation}
which is the USD-denominated GBP-based LIBOR. We note that the notation $x_{\pounds}$ and $y_{\pounds}$ stand for GBP-OIS and GBP-y-tenor, respectively. As mentioned earlier, we ignore the GBP-USD FX basis curve, for simplicity. However, this curve may be easily incorporated into the valuation through an intermediate curve-conversion factor process for GBP-OIS to GBP-USD FX basis. 

Next, we write the price $V^{x_{\$}\,y_{\pounds}}_{tT_i}$ at time $t\ge 0$ of the USD-denominated forward contract (written on the GBP-based LIBOR) that expires at $T_i\ge  t$. That is,
\begin{equation}
V^{x_{\$}\,y_{\pounds}}_{tT_i}=X^{\$\pounds}_0\frac{\delta_i}{h^{x_{\$}}_t}\E\left[h^{x_{\$}}_{T_i}\left(L^{x_{\$}\,y_{\pounds}}_{T_{i-1}}(T_{i-1},T_i)-K^{x_\$}\right)\,\vert\,\F_t\right],
\end{equation}
where $K^{x_\$}$ is the strike rate of the contract. By Eq. (\ref{FX-LIBOR}) and the tower property of conditional expectation, it then follows that
\begin{equation}
V^{x_{\$}\,y_{\pounds}}_{tT_i}=\frac{h^{y_{\pounds}}_t}{h^{x_{\$}}_t}P^{x_{\$}}_{tT_i}\,L^{y_{\pounds}}_t(T_{i-1},T_i)-K^{x_\$}\,P^{y_{\pounds}}_{tT_i}.
\end{equation}
Setting $V^{x_{\$}\,y_{\pounds}}_{tT_i}=0$ for all $t\in[0,T_i]$ gives the result (\ref{K-FX-LIBOR}), where Eqs (\ref{xy-conv-factor}), (\ref{mart-IBOR}), (\ref{xy-IBOR-mart}) and (\ref{fair-fwd-FX}) are used.
\end{proof}
\subsection{Inflation-linked foreign exchange hybrids}
Given that the relations for prices of inflation-linked and FX securities are available in the xy-approach, we can move on to the valuation of another hybrid financial instrument. We consider the price process of a contract that gives exposure to inflation in the domestic economy and is priced in a foreign currency. To answer this question, we take the example of a forward bet on inflation/deflation in the $j$-economy valued in units of the $i$-currency. This could be taking a bet at $t\in[0,T)$ on the growth in the value of the U. K. price index in EUR, $X^{ij}_T\,C^j_T/C^j_0$, at the fixed future date $T$.
\begin{Proposition} Let $(C^j_t)_{t\ge 0}$ be the $j$-economy price index process (\ref{PLP}) and $(X^{ij}_t)_{t\ge 0}$ the spot FX rate (\ref{spot-fx}). Consider the random payoff $V^{N_i}_T=X^{ij}_TC^j_T/C^j_0-K^{N_i}$, where $K^{N_i}$ is the nominal $i$-currency strike value, and $T$ is the fixed expiry date. The price process $(V^{N_i}_{tT})_{0\le t\le T}$ of the inflation-linked FX forward with cash flow $V^{N_i}_T$ is given by   
\begin{eqnarray}
V^{N_i}_{tT}=X^{ij}_0P^{N_i}_{tT}Q^{N_i R_j}_{tT}-K^{N_i}P^{N_i}_{tT},
\end{eqnarray}
where $(P^{N_i}_{tT})$ is the price process of the nominal discount bond in the $i$-economy, and where the conversion factor process $(Q^{N_i R_j}_{tT})_{0\le t\le T}$ is defined by
\begin{equation}
Q^{N_i R_j}_{tT}=\frac{\E\left[h^{R_j}_T\,\vert\,\F_t\right]}{\E\left[h^{N_i}_T\,\vert\,\F_t\right]}.
\end{equation}
The fair forward inflation-linked FX rate process $(F^{N_i R_j}_{tT})_{0\le t\le T}$ is given by
\begin{equation}\label{fair-fwd-ILFX}
F^{N_i R_j}_{tT}=F^{ij}_{tT}\,\frac{P^{N_j R_j}_{tT}}{P^{N_j}_{tT}}.
\end{equation}
\end{Proposition}
\begin{proof}
%
%
We begin with Proposition \ref{XY-gen-Prop}: Set $x=N_i$ and $y=R_j$ alongside $t=T$, and thereafter replace the pricing time $s$ with $t$. This gives,
\begin{equation}
H^{N_i R_j}_{tT}=\frac{1}{h^{N_i}_t}\E\left[h^{N_i}_T Q^{N_i R_j}_{TT}H^{R_j}_T\,\vert\,\F_t\right].
\end{equation}
For the real-economy random cash flow at time $T>0$ we set $H^{R_j}_T=X^{ij}_0-K^{N_i}/Q^{N_i R_j}_{TT}$, which is a quantity denominated in units of the $j$-real-economy. Now we calculate the price $V^{N_i}_{tT}$ at time $t\in[0,T]$ of the hybrid contract. We write $H^{N_i R_j}_{tT}=V^{N_i}_{tT}$ to emphasise that the value $V^{N_i}_{tT}$ at time $t\in[0,T]$ is given in nominal units of the economy with currency $i$. We have:
\begin{eqnarray}
V^{N_i}_{tT}&=&\frac{1}{h^{N_i}_t}\E\left[h^{N_i}_T\left(X^{ij}_T\,C^j_T/C^j_0-K^{N_i}\right)\,\big\vert\,\F_t\right]\\
\nn\\
&=&\frac{X^{ij}_0}{h^{N_i}_t}\E\left[h^{R_j}_T\,\vert\,\F_t\right]-P^{N_i}_{tT}K^{N_i},
\end{eqnarray}
where Eqs (\ref{PLP}) and (\ref{spot-fx}) are used.
This can be expressed in terms of the appropriate conversion factor process $(Q^{N_i R_j}_{tT})_{0\le t\le T}$. That is,
\begin{eqnarray}
V^{N_i}_{tT}=X^{ij}_0P^{N_i}_{tT}Q^{N_i R_j}_{tT}-P^{N_i}_{tT}K^{N_i},
\end{eqnarray}
where 
\begin{equation}
Q^{N_i R_j}_{tT}=\frac{\E\left[h^{R_j}_T\,\vert\,\F_t\right]}{\E\left[h^{N_i}_T\,\vert\,\F_t\right]}.
\end{equation}
Setting $V^{N_i}_{tT}=0$, for all $t\in[0,T]$, we obtain the fair inflation-linked FX forward process: 
\begin{equation}\label{fair-fwd-ILFX-old}
F^{N_i R_j}_{tT}=X^{ij}_0Q^{N_i R_j}_{tT}=F^{ij}_{tT}\,\frac{P^{N_j R_j}_{tT}}{P^{N_j}_{tT}},
\end{equation}
which concludes the proof.
\end{proof}
In summary, Equation (\ref{fair-fwd-ILFX}) states that the fair rate of an inflation-linked FX forward is given by $F^{ij}_{tT}/P^{N_j}_{tT}$ units of the bond $P^{N_j R_j}_{tT}$, which is linked to inflation in the (domestic) $j$-economy. In developed markets, the assets with price $F^{ij}_{tT}$, $P^{N_j}_{tT}$  and $P^{N_j R_j}_{tT}$, respectively, are (mostly) liquidly traded. The relation (\ref{fair-fwd-ILFX}) determines the consistent hedge for the $i$-currency inflation-linked FX forward in terms of the $j$-economy FX forward, the inflation-linked bond and the zero-coupon bond in the $j$-market. We thus have (i) the {\it consistent curve-conversion formula} $Q^{N_i R_j}_{tT}=Q^{ij}_{tT}Q^{N_j R_j}_{tT}$, linking inflation-indexed and FX securities, and (ii) the equivalent {\it consistent relation} (\ref{fair-fwd-ILFX}) between the inflation-indexed and FX forward rates.

\section{Conclusions}\label{Concl}

In this paper, a framework is developed that allows for the consistent pricing and hedging of financial assets, which depend on a spread between the rates their values accrue and are discounted at. Such a situation is manifest in fixed-income markets, in particular, where the return of instruments may accrue at one benchmark rate, e.g. LIBOR, and is discounted at another benchmark rate, e.g. the OIS rate. The paradigm for modelling the prices of tenor-based fixed-income products is the so-called multi-curve term structure framework. 
\

Although the approach we develop in this paper is applicable whenever spreads among different curves (term structures) need to be modelled, we consider fixed-income as the market within which we develop what we term {\it consistent valuation across curves}. We choose the modelling paradigm of pricing kernels to construct the consistent price systems that give rise to, and also rely on, the curve-conversion process that allows for no-arbitrage price conversions from one curve to another, as e.g. required in multi-curve interest rate modelling. This can be viewed as a kind of currency foreign-exchange analogy, and we draw several parallels with this view while we develop the {\it xy-approach}. 
\

After the introduction of the curve-dependent discounting systems, we produce the curve-conversion factor process that links cash-flows associated with different curves and hence gives rise to consistent prices of assets, which accrue value according to the {\it forecasting curve} and are discounted according to the {\it discounting curve}. The dual nature of the curve-conversion factor also allows for conversion of curves. The deduced across-curve pricing formula gives rise to the consistent set of numeraire assets and associated (risk-neutral) probability measures so as to avoid the introduction of arbitrage opportunities in a multi-curve market---or a `spread market'---see Section \ref{CCF}. The curve-conversion mechanism enables the introduction of tenor-based zero-coupon bonds without undermining the no-arbitrage requirement.
\

An intriguing by-product of the across-curve pricing kernel approach we develop is that it proposes consistent pricing relations for multi-curve systems in emerging markets where a derivatives market on one of the benchmarks, say the OIS system, is absent and needs to be estimated. For example, the liquidly traded tenor may be used to calibrate the pricing kernel model underlying the zero-coupon bond price system associated with the liquid tenor. Although a multi-curve interest rate system is available in emerging markets, it has an idiosyncratic and proprietary nature. Given an estimation methodology, one can apply the xy-approach as a standard to consistently price instruments in a multi-curve emerging market. We show that the across-curve valuation method is applicable in developed (liquid) markets as much as in emerging (less liquid) markets and that fixed-income products, such as forward rate agreements, may be understood and priced with the same ease in both types of markets. 
\

Recently, interest in so-called {\it rational models} has grown and the advantages of using this class of models to produce tractable interest-rate models, and extensions to the multi-curve setting, have been recognised. We develop generic pricing kernel models for across-curve valuation and show how rational multi-curve models, such as those of Cr\'epey et al. \cite{cmns} and Nguyen \& Seifried \cite{ns} are recovered within our xy-approach, and furthermore the linear-rational term structure models by Filipovi\'c et al. \cite{flt} may be generalised to a multi-curve environment. Moreover, important contributions have been made by several authors to produce multi-curve extensions of the Heath-Jarrow-Morton framework. We try to contribute to this research area by investigating the HJM-framework from the perspective of our across-curve valuation scheme therewith suggesting multi-curve HJM-models.
\

Finally we show how inflation-linked, currency-based, and fixed-income hybrid securities can be priced by applying our consistent across-curve valuation method using pricing kernels.
\\

\noindent {\bf Acknowledgments}.\,
The authors are grateful to participants of the Avior Quantitative Finance Research Seminar (Cape Town, March 2016), ACQuFRR \& RMB Masterclass on ``Catching up with Emerging Markets'' (Johannesburg, July 2016), CFE 2016 (Sevilla, December 2016), STM2016 Workshop, Institute of Statistical Mathematics, Japan (Tokyo, July 2016), JAFEE 2016 (Tokyo, August 2016), and of the Seminar at the Graduate School of International Corporate Strategy, Hitotsubashi University (Tokyo, March 2017) for comments and suggestions. Moreover, questions, comments and suggestions by Henrik Dam, Martino Grasselli, Matheus Grasselli and Erik Schl\"ogl have been very much appreciated and have led us to understand and develop our work on an across-curve valuation approach better and further. The authors thank two anonymous reviewers with comments which helped improve this paper.
\begin{appendix}
\section{No-arbitrage strategy for conversion of cash flows and curves}\label{App-ArbStrat-CC}

Let us first consider a simple arbitrage relationship for an economy with default-free and credit-risky interest rate curves, while assuming perfect market liquidity.  Assume that the $x$-curve is the default-free curve while the $y$-curve is one of a potential set of credit-risky curves. Consider the following simple strategy, at time 0:
\begin{enumerate}
	\item [(i)] Sell one unit of the numeraire asset in the $x$-market for $1/h^x_0$; and
	\item [(ii)] Buy one unit of the numeraire asset in the $y$-market which costs $1/h^y_0$,
\end{enumerate}
which costs zero to setup, i.e. $V_t=0$, since $h^x_0 = h^y_0 = 1$. Transaction (i) is equivalent to borrowing money via the $x$-market's money market, while (ii) is equivalent to a deposit into the $y$-market's money market. Then at any time $t > 0$, the value of this strategy will be 
\begin{equation*}
V_t = \frac{1}{h_t^y} - \frac{1}{h_t^x} , 
\end{equation*}
which will be greater than zero if the risky entity that holds the investment has not defaulted by that time. Therefore, this strategy does not allow for arbitrage, in general. Now, let us assume that one is able to mitigate all of the default-risk associated with the entity offering the $y$-market investment via appropriate collateralisation. In such a circumstance, the value of the strategy at any time $t > 0$ must equal zero, if we are to preclude arbitrage, otherwise one would be ensured of earning a cash flow equal to $V_t$ which would be greater than zero with certainty at any time $t > 0$. No arbitrage may be achieved by adjusting the $y$-market deposit by the ratio $h^y_t/h^x_t$. At any time $t > 0$, this ratio is merely the realised multiplicative spread between the discount factors realised in the $x$- and $y$-markets respectively.
\begin{Remark}
In currency modelling, the ratio $h^y_t/h^x_t$ models the spot exchange rate between the $x$- and $y$-currencies. In particular, 1 unit of $y$-currency may be exchanged for $h^y_t/h^x_t$ units of $x$-currency at time $t$.
\end{Remark}

Another relevant arbitrage relationship to consider involves a finite horizon loan and investment strategy. Maintaining the same assumptions as before, consider the same strategy as before at time 0, and then do the following at some time $t  \in (0, T)$:
\begin{enumerate}
	\item [(i)] Sell $1/h_t^x$ units of the $x$-market $T$-maturity bond for $P^x_{tT}$; and
	\item [(ii)] Buy $1/h_t^y$ units of the $y$-market $T$-maturity bond for $P^y_{tT}$,
\end{enumerate}
which again costs zero to setup at time 0, as before, and terminates at time $t$ when the money market loan and deposit is transferred to fixed horizon alternatives. Now, at any time $s \in [0,t)$, the same arguments apply as before while at time $t$ the value of this strategy will be
\begin{equation*}
V_t = \frac{1}{h_t^y P^y_{tT}} - \frac{1}{h_t^x P_{tT}^x} , 
\end{equation*}
which again does not permit an arbitrage opportunity, due to the credit risk associated with the investment leg of the strategy. If we again invoke collateralisation of the investment leg of the strategy, then arbitrage is precluded at: (a) all times $s \in [0,t)$  by adjusting the $y$-market deposit by the ratio $h^y_t/h^x_t$; and (b) at time $t$ by adjusting the $y$-market fixed term deposit by the ratio $h^y_t P_{tT}^y/h^x_t P_{tT}^x$. If these adjustments are not enforced post collateralisation, then one would be ensured of a risk-free profit equal to $V_t$ for all times $t \in (0,T]$.
\begin{Remark}
In currency modelling, the ratio $h^y_t P_{tT}^y/h^x_t P_{tT}^x$ models the forward exchange rate between the $x$- and $y$-currencies. In particular, one can agree at time $t$ to exchange 1 unit of $y$-currency for $h^y_t P_{tT}^y/h^x_t P_{tT}^x$ units of $x$-currency at time $T \geq t$.
\end{Remark}

\section{Consistent changes of numeraire and measure}\label{App-B}

Here we discuss changes-of-measure, numeraire assets, martingales and therefore no-arbitrage within the xy-formalism. The curve-conversion factor process (\ref{xy-conv-factor}) induces the changes-of-measure between all introduced $y$-markets (or $y$-curves). In particular, it governs no-arbitrage across all distinct markets associated with the economy under consideration. To demonstrate this, we consider Proposition \ref{XY-gen-Prop}, along with an asset with a spot-defined future cash flow $H^y_{T}$ and deduce that the value of such an asset in the $x$-market is
\begin{equation}
H^{xy}_{tT}=\frac{1}{h^x_t}\E\left[h^y_T\, H^y_{T}\,\vert\,\F_t\right],
\end{equation}
for $t\in[0,T]$. For each of the markets $z = x,y$, we introduce a change-of-measure density martingale $(m^z_t)_{0\le t\le T}$ which changes measure from the real-world measure $\PR$ to the equivalent (risk-neutral) measure $\Q_z$, along with the $z$-discount factor $(D^z_t)$ such that $1/D^z_t$ is the natural numeraire under $\Q_z$. This also means that the $z$-market's pricing kernel may be written as $h^z_t = D^z_t m^z_t$. The price process $(H^{xy}_{tT})_{0\le t\le T}$ can now be expressed, equivalently, in terms of (a) the $\Q_x$ risk-neutral measure and (b) the $\Q_y$ risk-neutral measure:
 \begin{equation}\label{Hxy-spot-measure-changes}
 H^{xy}_{tT} = \frac{1}{D^x_t}\E^{\Q_x}\left[D^x_T Q^{xy}_{TT} H^y_{T}\,\vert\,\F_t\right] = \frac{m^y_t}{m^x_t}\frac{1}{D^x_t}\E^{\Q_y}\left[D^y_T H^y_{T}\,\vert\,\F_t\right] = Q^{xy}_{tt} H^y_t,
\end{equation}
where we emphasise that $(D^x_t H^{xy}_{tT})_{0 \leq t \leq T}$ and $(D^y_t H^y_{t})_{0 \leq t \leq T}$ are $\Q_x$- and $\Q_y$-martingales respectively, by construction. Moreover, we may change measure from $\Q_z$ to the $T$-forward measure $\Q^T_z$ via the Radon-Nikodym derivative
\begin{equation}
\frac{\rd \Q^T_z}{\rd \Q_z} = \frac{D^z_T P^z_{TT}}{D^z_t P^z_{tT}} \,,
\end{equation}
which acts on $\F_T$ given information $\F_t$ up to time $t$, and therefore we may now express the price process $(H^{xy}_{tT})_{0\le t\le T}$ equivalently, in terms of (a) the $x$-market $T$-forward measure and (b) the $y$-market $T$-forward measure:
\begin{equation}\label{Hxy-fwd-measure-changes}
H^{xy}_{tT} = P^x_{tT} \E^{\Q^T_x}\left[ Q^{xy}_{TT} H^y_{T}\,\vert\,\F_t\right]=\frac{m^y_t D^y_t P^y_{tT}}{m^x_t D^x_t}\E^{\Q^T_y}\left[ H^y_{T}\,\vert\,\F_t\right] = P^x_{tT} Q^{xy}_{tT} \E^{\Q^T_y}\left[ H^y_{T}\,\vert\,\F_t\right] = P^x_{tT} Q^{xy}_{tT} \frac{H^y_t}{P^y_{tT}},
\end{equation}
where $(H^{xy}_{tT}/P^x_{tT})_{0 \leq t \leq T}$ and $(H^y_{t}/P^y_{tT})_{0 \leq t \leq T}$ are $\Q^T_x$- and $\Q^T_y$-martingales respectively. Equations (\ref{Hxy-spot-measure-changes}) and (\ref{Hxy-fwd-measure-changes}) clearly demonstrate the role of the curve-conversion factor process in changing measure within and across markets. Furthermore, the price process' martingale property is preserved across markets (and curves), with the curve-conversion factor process again enabling this property. The xy-approach precludes arbitrage within and across different markets (and curves).

\section{Bootstrapping of initial term structures}\label{App-C}

\subsection{Emerging markets}\label{App-Boot-EM}
In an emerging market, one would have the following initial data: (a) the $y$-tenored spot IBOR $L^y_0(0,T_1)$; (b) a set of fair FRA rates $\{K_0^{yy}(T_1,T_2),K_0^{yy}(T_2,T_3), \ldots, K_0^{yy}(T_{n-1},T_n)\}$; and (c) a set of fair IRS rates $\{S_0^{yy}(0,T_{n+1}), S_0^{yy}(0,T_{n+2}), \ldots, S_0^{yy}(0,T_{n+m})\}$. Using this data, one may construct the initial $y$-ZCB system by the relations
\begin{eqnarray}\label{yy-EM-bootstrap}
P^{y}_{0T_1} &=& \frac{1}{1+L^y_0(0,T_1) \delta_1} , \nn \\
P^{y}_{0T_i}  &=& \frac{P^{y}_{0T_{i-1}}}{1 + K_0^{yy}(T_{i-1},T_i) \delta_i} , \nn \\
P^{y}_{0T_{n+j}} &=&  \frac{1-S_0^{yy}(0,T_{n+j})\sum_{k=1}^{n+j-1} \delta_k  P^{y}_{0T_k} }{1 + \delta_{n+j}S^{yy}_0(0,T_{n+j}) }  ,
\end{eqnarray}
for $i \in \{2,3,\ldots,n\}$ and $j \in \{1,2,\ldots,m\}$. In general, one will have to make use of a suitable numerical bootstrapping technique to extend the $y$-ZCB system from the longest FRA maturity to the set of IRS maturities. These results are all consistent with a classical single-curve interest rate framework.
\subsection{Developed markets}\label{App-Boot-DM}
In a developed market, one would have the following initial data: (a) the $y$-tenored spot IBOR $L^{xy}_0(0,T_1)$; (b) a set of fair FRA rates $\{K_0^{xy}(T_1,T_2),K_0^{xy}(T_2,T_3), \ldots, K_0^{xy}(T_{n-1},T_n)\}$; and (c) a set of fair IRS rates $\{S_0^{xy}(0,T_{n+1}), S_0^{xy}(0,T_{n+2}), \ldots, S_0^{xy}(0,T_{n+m})\}$. Using this data, one may construct the initial $y$-ZCB system by the relations
\begin{eqnarray}
P^{y}_{0T_1} &=& 1 - \delta_i P^{x}_{0T_1} L^{xy}_0(0,T_1), \nn \\
P^{y}_{0T_i}  &=& P^{y}_{0T_{i-1}} - \delta_i P^{x}_{0T_i} K^{xy}_0(T_{i-1},T_i), \nn \\
P^{y}_{0T_{n+j}}  &=&  1-S^{xy}_0(0,T_{n+j}) \sum_{k=1}^{n+j} \delta_k P^{x}_{0T_k},
\end{eqnarray}
for $i \in \{2,3,\ldots,n\}$ and $j \in \{1,2,\ldots,m\}$. In general, one will have to make use of a suitable numerical bootstrapping technique to extend the $y$-ZCB system from the longest FRA maturity to the set of IRS maturities. Interestingly, but not surprisingly, since $P^{xy}_{st} = (h^{y}_s/h^{x}_s) P^{y}_{st}$ for $0 \leq s \leq t$, it follows that
\begin{equation}\label{Py-Pxy-relation}
L^{y}_t(T_{i-1},T_i) = \frac{1}{\delta_i}\left( \frac{P^{y}_{tT_{i-1}}}{P^{y}_{tT_i}} - 1\right) = \frac{1}{\delta_i}\left( \frac{P^{xy}_{tT_{i-1}}}{P^{xy}_{tT_i}} - 1\right) ,
\end{equation}
and therefore $P^{xy}_{0t} = P^{y}_{0t}$ for $t \geq 0$.
\begin{Remark}
Market practitioners may choose to construct a market-implied $\overline{y}$-ZCB system, which we will denote by $\{P^{\overline{y}}_{0t}\}_{t \geq 0}$, as follows
\begin{eqnarray}
P^{\overline{y}}_{0T_1} &=& \frac{1}{1+L^{xy}_0(0,T_1) \delta_1} , \nn \\
P^{\overline{y}}_{0T_i}  &=& \frac{P^{\overline{y}}_{0T_{i-1}}}{1 + L_0^{xy}(T_{i-1},T_i) \delta_i} ,\nn \\
P^{\overline{y}}_{0T_{n+j}} &=& \frac{1-\sum_{k=1}^{n+j-1} \delta_k S_0^{xy}(0,T_{n+j}) P^{\overline{y}}_{0T_k} }{1 + \delta_{n+j}S^{xy}_0(0,T_{n+j}) }  ,
\end{eqnarray}
for $i \in \{2,3,\ldots,n\}$ and $j \in \{1,2,\ldots,m\}$, using the same initial data, available in a \textit{developed market}, as before. This is indeed what is currently done in practice, with no attention paid to the fact that assumption (\ref{DM-martingale-assumption}) confounds the true forward IBOR process with a martingale adjustment. Put differently, the resultant $\overline{y}$-ZCB system is dependent on the $x$-ZCB system via the conversion factor $Q^{xy}_{\cdot \cdot}$.
\end{Remark}

\end{appendix}


\end{document}